\documentclass[12pt,a4paper]{amsart}

\usepackage[scale={.7875, .814},marginratio={1:1, 1:1}]{geometry}

\title{Complexity Theory for Operators in Analysis}
\thanks{A short preliminary version of this work was presented at the 42nd ACM Symposium on Theory of Computing (STOC~2010).}
\author[Kawamura]{Akitoshi Kawamura}
\author[Cook]{Stephen Cook}
\address[Kawamura]{Department of Computer Science\\
         University of Tokyo}
\address[Cook]{Department of Computer Science\\
         University of Toronto}

\usepackage[font=small,margin=18pt,labelfont=sc,labelsep=period]{caption}
\usepackage[all]{xy}
\usepackage{amsmath}
\usepackage{amssymb}
\usepackage{amsthm}
\usepackage[scaled]{helvet}
\usepackage[dvipdfm]{graphicx}
\usepackage{stmaryrd}

\newtheorem{theorem}{Theorem}[section]
\newtheorem{lemma}[theorem]{Lemma}
\newtheorem{corollary}[theorem]{Corollary}
\theoremstyle{definition}
\newtheorem{definition}[theorem]{Definition}
\theoremstyle{remark}
\newtheorem{example}[theorem]{Example}

\usepackage[square,comma,numbers,sort&compress]{natbib}

\newcommand{\classfont}[1]{\text{\sffamily \upshape #1}}
\newcommand{\classtwofont}[1]{\text{\bfseries \sffamily \upshape #1}}

\newcommand{\classP}{\classfont{P}}
\newcommand{\classPtwo}{\classtwofont{P}}
\newcommand{\classFP}{\classfont{FP}}
\newcommand{\classFPtwo}{\classtwofont{FP}}
\newcommand{\classNP}{\classfont{NP}}
\newcommand{\classNPtwo}{\classtwofont{NP}}
\newcommand{\classPSPACE}{\classfont{PSPACE}}
\newcommand{\classPSPACEtwo}{\classtwofont{PSPACE}}
\newcommand{\classFPSPACE}{\classfont{FPSPACE}}
\newcommand{\classFPSPACEtwo}{\classtwofont{FPSPACE}}

\newcommand{\redonem}{\leq _{\mathrm{mF}} ^1}
\newcommand{\redonep}{\leq _{\mathrm{m}} ^1}
\newcommand{\redoneT}{\leq _{\mathrm{T}} ^1}
\newcommand{\redtwom}{\leq _{\mathrm{mF}} ^2}
\newcommand{\redtwop}{\leq _{\mathrm{m}} ^2}
\newcommand{\redtwoT}{\leq _{\mathrm{W}} ^2}

\newcommand{\OpLipIVP}{\mathit{LipIVP}}
\newcommand{\OpCH}{\mathit{CH}}
\newcommand{\dist}{\mathit{dist}}
\newcommand{\OpApply}{\mathit{Apply}}

\newcommand{\id}{\mathrm{id}}

\newcommand{\lsem}{\llbracket}
\newcommand{\rsem}{\rrbracket}

\newcommand{\rhoreal}{\rho _\Rset}
\newcommand{\rhosd}{\rho _{\mathrm{sd}}}
\newcommand{\psiclosed}{\psi _{\circledcirc}}
\newcommand{\deltabox}{\delta _{\square}}
\newcommand{\deltalip}{\delta _{\square \mathrm L}}

\newcommand{\restrictcodom}[1]{\mathord| ^{#1}}

\newcommand{\classC}{\mathrm C}
\newcommand{\classCL}{\mathrm{CL}}

\newcommand{\lit}[1]{\texttt{\textup{#1}}}

\makeatletter
\def\atmosttextstyle#1{\mathchoice{\textstyle#1}{\textstyle#1}{\scriptstyle#1}{\scriptscriptstyle#1}}
\def\notsobig#1{{
  \def\t@mp{#1}
  \xdef#1{\mathop{\atmosttextstyle{\t@mp}}}
}}
\makeatother
\notsobig\sum
\notsobig\prod

\newcommand{\Nset}{\mathbb N}

\newcommand{\Rset}{\mathbb R}

\newcommand{\Bset}{\mathbf{Reg}}
\newcommand{\Pset}{\mathbf{Pred}}

\newcommand{\Dset}{\mathbf D}

\DeclareMathOperator{\dom}{dom}

\newcommand{\tcolon}{\colon\:}
\newcommand{\tto}{\mathbin{\to}}

\begin{document}

\begin{abstract}
We propose an extension of the framework for 
discussing the computational complexity of 
problems involving uncountably many objects, 
such as real numbers, sets and functions, 
that can be represented only 
through approximation. 
The key idea is to use (a certain class of) string functions
as names representing these objects. 
These are more expressive than infinite sequences, 
which served as names in prior work that 
formulated complexity in more restricted settings. 
An advantage of using string functions
is that we can define 
their \emph{size} 
in the way inspired by 
higher-type complexity theory. 
This enables us to talk about computation on string functions 
whose time or space is bounded polynomially in the input size, 
giving rise to more general analogues of the classes 
$\classP$, $\classNP$, and $\classPSPACE$. 
We also define $\classNP$- and $\classPSPACE$-completeness
under suitable many-one reductions. 

Because our framework separates machine computation and semantics, 
it can be applied to problems on sets of interest in analysis
once we specify a suitable representation (encoding). 
As prototype applications, 
we consider the complexity of functions (operators) on
real numbers, real sets, and real functions. 
For example, the task of numerical algorithms for 
solving a certain class of differential equations 
is naturally viewed as an operator taking real functions to real functions.  
As there was no complexity theory for operators, 
previous results only stated how complex the solution can be. 
We now reformulate them and show that the operator itself
is polynomial-space complete. 
\end{abstract}

\maketitle

\section{Introduction}

\emph{Computable Analysis} 
\cite{weihrauch00:_comput_analy, brattka08:_tutor_comput_analy}
studies problems 
involving real numbers
from the viewpoint of computability. 
Elements of uncountable sets (such as real numbers) are
represented
by infinite sequences of approximations 
and processed by Turing machines. 
This framework is applicable not only to the real numbers
but also with great generality 
to other spaces arising naturally in mathematical analysis. 
There is a unified way to discuss computability of 
real functions, sets of real numbers, 
operators taking real functions as inputs, and so on. 

In contrast, 
the application of this approach to 
computational complexity 
has been limited in generality. 
For example, 
although 
there is a widely accepted notion of polynomial-time computable
real functions $f \tcolon [0, 1] \tto \Rset$ on the compact interval 
that has been studied extensively~%
\cite{ko98:_polyn_time_comput_analy}, 
the same approach does not give a nice class of 
real functions on $\Rset$. 
Most of the complexity results in computable analysis to date 
(with a few exceptions~%
\cite{hoover90:_feasib_real_funct_and_arith_circuit, 
      takeuti01:_effec_fixed_point_theor_over, 
      weihrauch03:_comput_compl_comput_metric_spaces})
are essentially limited to the complexity of either 
real functions with compact domain, or 
of bounded subsets of $\Rset$. 
They do not address the complexity of, say, 
an operator~$F$ that takes 
real functions $f \tcolon [0, 1] \tto \Rset$ to another real function $F (f)$. 
There are many positive and negative results~%
\cite{ko91:_comput_compl_of_real_funct}
about such operators, 
but typically they are stated in the form 
\begin{quote}
if $f$ is in the complexity class $X$, 
then $F (f)$ is in complexity class $Y$, and \\
there is $f$ in complexity class $X$ 
such that $F (f)$ is hard for $Z$. 
\end{quote}
More direct
statements would be 
the ``uniform'' or ``constructive'' form
\begin{quote}
the operator $F$ is in class $\mathcal Y$, and \\
the operator $F$ is $\mathcal Z$-hard, 
\end{quote}
where $\mathcal Y$ and $\mathcal Z$ are 
the ``higher-order versions'' of $Y$ and $Z$. 
At the level of computability, 
it is common to ask, as soon as we see a non-uniform result, 
whether it can be made uniform. 
For complexity, we cannot even ask this question 
because we do not know how to formulate $\mathcal Y$ and $\mathcal Z$. 
This limitation has been widely recognized; 
see, for example, 
\cite[pp.\,57--58]{ko91:_comput_compl_of_real_funct}, 
\cite{weihrauch03:_comput_compl_comput_metric_spaces}, 
and 
\cite[p.\,484]{brattka08:_tutor_comput_analy}. 

To address this problem, 
we start with the observation (Section~\ref{section: TTE}) that 
the aforementioned limitation has to do with the fact that 
traditional formulations of computable analysis 
do not take into account the ``size'' of 
the infinite sequences given to the machine as input. 
We then propose (Section~\ref{section: main}) 
an extension on the machine model
by replacing infinite sequences
by what we call \emph{regular functions}
on strings. 
An advantage of using these functions is 
that we can define their \emph{size} 
in the way suggested by type-two complexity theory~%
\cite{mehlhorn76:_polyn_and_abstr_subrec_class, kapron96:_new_charac_of_type_feasib}. 
This enables us to measure the 
growth of running time (or space) in terms of the input size---%
exactly what we do in the usual (type-one) complexity theory. 
We thus obtain the complexity classes analogous to
$\classP$, $\classNP$, $\classPSPACE$ 
(and function classes $\classFP$ and $\classFPSPACE$)
by bounding the time or space by 
\emph{second-order polynomials}
in the input size. 
Analogues of many-one reductions and 
$\classNP$- and $\classPSPACE$-hardness 
will also be introduced. 

We apply this framework 
to a few specific problems in analysis 
by using suitable representations of 
real numbers, real sets, and real functions (Section~\ref{section: applications}). 
For real numbers, 
the induced complexity notions turn out 
to be equivalent to what has been studied
by Ko--Friedman~%
\cite{ko82:_comput_compl_of_real_funct} and 
Hoover~%
\cite{hoover90:_feasib_real_funct_and_arith_circuit}. 
For sets and functions, 
our approach seems to be the first to 
provide complexity notions in a unified manner. 
This is of particular interest, 
because many numerical problems in the real world are 
naturally formulated as operators taking sets or functions. 
For example, 
consider the operator~$F$ that finds the solution $F (f)$ 
of the differential equation (of a certain class) given by a function~$f$. 
As mentioned above, 
the existing non-uniform results~%
\cite{ko83:_comput_compl_of_ordin_differ_equat, 
      kawamura10:_lipsc_contin_ordin_differ_equat}
only tell us
\emph{how complex the solution $F (f)$ can be when $f$ is easy}; 
precisely, they say that if $f$ is polynomial-time computable, 
$F (f)$ is polynomial-space computable and can be polynomial-space hard. 
But the practical concern for numerical analysis would be 
\emph{how hard it is to compute $F$} 
(i.e., to compute $F (f)$ given $f$). 
We formulate and prove the first result of this kind: 
$F$ itself is a polynomial-space complete operator. 
Our contribution is in introducing the framework making such
formulations possible. 
The technically hard parts of the proofs of the specific results 
are already done in the proofs of the non-uniform versions, 
and all we need to do is to check that they uniformize in our sense. 
The original non-uniform versions are now
corollaries of the uniform statements. 

\subsection*{Notation and terminology}
\label{section: notation}

A \emph{multi-valued function}
(or \emph{multi-function}) $F$ from a set~$X$ to a set~$Y$ 
is formally a subset of $X \times Y$. 
For $x \in X$, we write $F [x]$ for the
set of $y \in Y$ such that $(x, y)$ belongs to this subset. 
These $y$ are the ``allowable outputs'' on input~$x$. 
We denote by $\dom F$ the set of $x \in X$ for which $F [x]$ is nonempty. 
When $F [x]$ is a singleton, 
its unique element is denoted by $F (x)$, as usual. 
If $F [x]$ is a singleton for all $x \in \dom F$, 
we say that $F$ is a \emph{partial function}. 
When in addition $\dom F = X$, we say that $F$ is a 
\emph{total function}, or simply a \emph{function}. 

Like some authors~%
\cite{goldreich, wegener}, 
we regard computational tasks (problems)
as multi-functions. 
The classes $\classFP$ and $\classFPSPACE$ 
consist of multi-functions from strings to strings 
computed by a machine 
whose time/space is polynomially bounded. 
Here, computing a multi-function is to be interpreted according to 
the ``allowable outputs'' semantics mentioned above: 
A machine is said to compute $F$ if, on any input $x \in \dom F$, 
it outputs \emph{some} element of $F [x]$. 
The classes $\classFPtwo$ and $\classFPSPACEtwo$ that we will define later 
will also consist of multi-functions. 

Note that we do not care what happens on inputs outside $\dom F$, 
unlike some authors who require that such inputs be rejected explicitly. 
Thus a multi-function can be easy to compute while having a nasty domain. 
We also note, however, that allowing $\dom F$ to be smaller than $X$ is 
not so important in the context of time- or space-bounded computation, 
because a machine that runs past the bound for some inputs
can be modified so that it keeps track of the time 
and outputs an error message when it has run out of time or space. 

Throughout the paper, 
$\varSigma ^*$ denotes the set of finite strings 
over the alphabet $\varSigma$. 
We will tacitly assume, 
depending on contexts, 
that $\varSigma = \{0, 1\}$ or 
that $\varSigma$ contains all symbols 
appearing in the discussion. 

Since our applications mainly involve real numbers, 
it will be convenient to fix a 
dense subset of $\Rset$ and its encoding. 
For each $n \in \Nset$, 
let $\Dset _n$ denote 
the set of strings of the form 
\begin{equation}
\label{equation: names of dyadic numbers}
 s x / 1 \mkern-3mu \underbrace{00\dots0} _n, 
\end{equation}
where $s \in \{\mathord+, \mathord-\}$ and $x \in \{0, 1\} ^*$. 
Let $\Dset = \bigcup _{n \in \Nset} \Dset _n$. 
A string in $\Dset$ 
\emph{encodes} 
a number in the obvious sense---namely, 
read \eqref{equation: names of dyadic numbers} as 
a fraction whose numerator and denominator are integers 
written in binary with leading zeros allowed. 
We write $\lsem u \rsem$ for the 
number encoded by $u \in \Dset$. 
The numbers that can be encoded in this way are called 
\emph{dyadic} numbers. 

\section{Type-Two Theory of Effectivity}
\label{section: TTE}

There are several equivalent formulations for Computable Analysis. 
One powerful framework is 
Weihrauch's Type-Two Theory of Effectivity (TTE)
\cite{weihrauch00:_comput_analy, 
      brattka08:_tutor_comput_analy}. 
In this section, 
we brief\textcompwordmark ly introduce the
infinite sequence model of TTE 
and discuss some difficulties in dealing with complexity, 
which motivate our modification 
in Section~\ref{section: main}. 

\subsection{Computability}
\label{subsection: TTE computability}

In the usual computability theory, 
we use some machine model that
computes functions from $\varSigma ^*$ to $\varSigma ^*$. 
To discuss computation on other sets~$X$, 
we specify an \emph{encoding} of $X$---that is, 
a rule for interpreting an element of $\varSigma ^*$ as an element of $X$. 

But we want to deal with 
uncountable sets, such as the set $\Rset$ of real numbers. 
Since the countable set $\varSigma ^*$ cannot encode them, 
TTE uses the set $\varSigma ^\Nset$ of \emph{infinite sequences} instead. 

Computability of partial functions 
from $\varSigma ^\Nset$ to $\varSigma ^\Nset$ is defined 
using Turing machines. 
The machine has an input tape, an output tape, and a work tape, 
each of which is infinite to the right. 
We also assume that the output tape is one-way; that is, 
the only instruction for the output tape is 
``write $a \in \varSigma$ in the current cell and move the head to the right''.
The difference from the usual setting 
is in the convention by which the machine reads the input and delivers the output. 
The input is now an infinite string $a _0 a _1 \ldots \in \varSigma ^\Nset$, 
and is written on the input tape before the computation starts
(with the tape head at the leftmost cell). 
We say the machine outputs 
an infinite string $b _0 b _1 \ldots \in \varSigma ^\Nset$
if it never halts and writes the string indefinitely on the output tape
(that is, 
for each $n \in \Nset$, it eventually writes $b _0 \ldots b _{n - 1}$ 
into the first $n$ cells). 
This defines a class of 
(possibly partial) computable functions (without any time or space bound)
from $\varSigma ^\Nset$ to $\varSigma ^\Nset$. 
The definition can be extended to multi-functions~$A$: 
we say that a machine $M$ computes $A$ if $M$, on any input $\varphi \in \dom A$, 
always outputs some element of $A [\varphi]$. 

A 
\emph{representation}
$\gamma$ of a set~$X$ is formally 
a partial function from $\varSigma ^\Nset$ to $X$ which is 
surjective---that is, 
for each $x \in X$, 
there is at least one $\varphi \in \varSigma ^\Nset$ with $\gamma (\varphi) = x$. 
We say that $\varphi$ is a 
\emph{$\gamma$-name}
of $x$. 
Computability of multi-functions on represented sets is defined as follows. 

\begin{definition}
 \label{definition: realization}
Let $\gamma$ and $\delta$ be representations of sets $X$ and $Y$, respectively. 
We say that a machine 
\emph{$(\gamma, \delta)$-computes} 
a multi-function $A$ from $X$ to $Y$ if it computes 
the multi-function $\delta ^{-1} \circ A \circ \gamma$ given by 
\begin{equation}
 (\delta ^{-1} \circ A \circ \gamma) [\varphi] 
=
\begin{cases}
 \{\, \psi \in \dom \delta : \delta (\psi) \in A [\gamma (\varphi)] \,\}
&
 \text{if} \ \varphi \in \dom \gamma, 
\\
 \emptyset
& 
 \text{otherwise}. 
\end{cases}
\end{equation}
\end{definition}

In other words, 
whenever the machine is given a $\gamma$-name of an element $x \in \dom A$, 
it must output some $\delta$-name of some element of $A [x]$ 
(Figure~\ref{figure: realization}). 
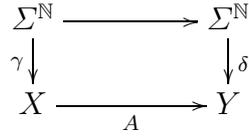
\begin{figure}[t]
\begin{equation*}
  \xymatrix@C=50pt@R=16pt{
   \varSigma ^\Nset \ar[r] \ar[d] _\gamma & \varSigma ^\Nset \ar[d] ^\delta  \\
   X \ar[r] _A & Y 
  }
\end{equation*}
 \caption{$(\gamma, \delta)$-computing a multi-function~$A$.}
 \label{figure: realization}
\end{figure}

As an example, 
we define a representation~$\rhoreal$ of the set $\Rset$ of real numbers 
by saying that 
an infinite string $\varphi \in \varSigma ^\Nset$ 
is a $\rhoreal$-name of $x \in \Rset$ 
if $\varphi$ is of the form $u _0 \# u _1 \# u _2 \# \ldots$ 
(where $\#$ is a delimiter symbol not appearing in the $u _i$) such that 
$u _i \in \Dset$ and 
$\lvert \lsem u _i \rsem - x \rvert < 2 ^{-i}$ for each $i \in \Nset$. 
Thus a real number is specified by 
a list of rational numbers converging to it. 
It turns out that $\rhoreal$ is a natural representation
with which to discuss computability of real functions. 
In particular, $\rhoreal$ is \emph{admissible}
with respect to the usual topology of $\Rset$~%
\cite[Lemma~4.1.6]{weihrauch00:_comput_analy}. 

To deal with functions of two arguments, 
we define, for representations $\gamma$ and $\delta$ of sets $X$ and $Y$, 
a representation $[\gamma, \delta]$ of $X \times Y$ by $
[\gamma, \delta] (a _0 b _0 a _1 b _1 \ldots) = 
(\gamma (a _0 a _1 \ldots), \delta (b _0 b _1 \ldots))
$. 

\begin{example}
 \label{example: addition}
Addition $\mathord+ \tcolon \Rset \times \Rset \tto \Rset$ is 
$([\rhoreal, \rhoreal], \rhoreal)$-computable. 
For suppose that we are given names 
$\varphi = u _0 \# u _1 \# \ldots$ and $\psi = v _0 \# v _1 \# \ldots$ 
of real numbers $s$ and $t$. 
An approximation of $s + t$ with precision $2 ^{-m}$, for each $m$, 
is given by $\lsem u _{m + 1} \rsem + \lsem v _{m + 1} \rsem$. 
\end{example}

\begin{example}
 \label{example: multiplication}
Multiplication $\mathord\times \tcolon \Rset \times \Rset \tto \Rset$ is 
$([\rhoreal, \rhoreal], \rhoreal)$-computable. 
Given
names $\varphi = u _0 \# u _1 \# \ldots$ and $\psi = v _0 \# v _1 \# \ldots$ 
of real numbers $s$ and $t$, 
let $k = \max \{ \lvert u _0 \rvert, \lvert v _0 \rvert \}$. 
Since $\lsem u _0 \rsem$ and $\lsem v _0 \rsem$ are near $s$ and $t$, 
and it takes more than $k$ digits to encode 
a number with absolute value $\geq 2 ^k$, 
we have $\lvert s \rvert$, $\lvert t \rvert < 2 ^k$. 
Hence, $s \times t$ is approximated 
with precision $2 ^{-m}$ by 
$\lsem u _{m + k + 1} \rsem \cdot \lsem v _{m + k + 1} \rsem$.
\end{example}

A good thing about the TTE formulation is that, 
by using suitable representations, 
we can discuss computation on many other sets. 
There are often standard ways to obtain
representations of higher-type objects such as sets and functions. 
For example, since we have agreed on the representations~$\rhoreal$ of $\Rset$, 
we can introduce a canonical representation of 
the set $\classC [\Rset]$ of continuous real functions, 
and there are reasons to believe that this 
is the ``right'' representation~%
\cite[Chapter~3]{weihrauch00:_comput_analy}. 

\subsection{Complexity}
\label{subsection: TTE complexity}

Now we start putting time bounds. 
This means requiring that the $n$th prefix of the output be delivered
within time bounded polynomially in $n$ (and independently of $\varphi$): 

\begin{definition}
\label{definition: polynomial time stream machine}
A machine~$M$ runs in 
\emph{polynomial time}
if there is a polynomial~$p$ such that 
for all $\varphi \in \varSigma ^\Nset$ and $n \in \Nset$, 
the machine $M$ on input $\varphi$ 
finishes writing the first $n$ symbols of the output 
within $p (n)$ steps. 
Define \emph{polynomial space} analogously by counting the number of visited cells on all (input, work and output) tapes. 
\end{definition}

Can we use this notion 
to define polynomial-time computability of, say, a real function? 

\subsubsection{Representations must be chosen carefully}
\label{subsubsection: TTE complexity compact}

A little thought shows
that the simple combination of 
Definition~\ref{definition: polynomial time stream machine} 
and the representation~$\rhoreal$ is useless~%
\cite[Examples 7.2.1, 7.2.3]{weihrauch00:_comput_analy}. 
On the one hand, the machine~$M$ could ``cheat'' by 
writing redundant $\rhoreal$-names: 
By writing $+10000/100000$ instead of $+1 / 10$ 
it gets more time to compute the next approximation. 
On the other hand, 
the machine may suffer by receiving redundant names as input, 
such as the one in which the first approximation is too long to even read in time. 

This motivates 
the use of 
\emph{signed digit representation}
$\rhosd$ of $\Rset$
\cite[Definition 7.2.4]{weihrauch00:_comput_analy} 
defined as follows, forbidding redundancy carefully: 
$\dom \rhosd$ consists of sequences $\varphi \in \varSigma ^\Nset$ of form $
a _k \ldots a _1 a _0 \bullet a _{-1} a _{-2} \ldots 
$ for some $k$, where each $a _i$ is either $0$, $1$ or $-1$, 
such that $k = 0$ or 
$(a _k, a _{k - 1}) \in \{(1, 0), (1, 1), (-1, 0), (-1, -1)\}$; 
if this is the case, set $
  \rhosd (\varphi) = \sum _{i = -\infty} ^k a _i \cdot 2 ^i 
$. 
Thus we read the digit sequence as a binary expansion of a real number
(with decimal point $\bullet$) 
with digits $0$, $1$ and $-1$; 
we forbid certain patterns in the first two digits of the integer part
in order to exclude redundancy. 
(See \cite[Example~2.1.4.7]{weihrauch00:_comput_analy} for the reason
why the usual binary expansion without the ``$-1$'' symbol does not work.)

Let $\rhosd | ^{[0, 1]}$ denote the restriction of $\rhosd$ to 
(infinite sequences representing) real numbers in $[0, 1]$. 
By Definition~\ref{definition: polynomial time stream machine}, 
we know what it means for a real function $f \tcolon [0, 1] \to \Rset$ to be
\emph{polynomial-time $(\rhosd | ^{[0, 1]}, \rhosd)$-computable}. 
This notion turns out to be robust and 
equivalent to the widely accepted 
polynomial-time computability of Ko and Friedman~%
\cite{ko82:_comput_compl_of_real_funct}, 
so we will drop the prefix ``$(\rhosd | ^{[0, 1]}, \rhosd)$'' from now on. 
The same goes for \emph{polynomial-space} computability, 
and for functions on compact intervals or rectangles instead of $[0, 1]$
(use the pairing function as in Examples 
\ref{example: addition} and \ref{example: multiplication}). 
It is routine to verify that, for example, 
addition and multiplication 
$\mathord+$, $\mathord \times \tcolon [0, 1] \times [0, 1] \tto \Rset$ are 
polynomial-time computable. 
For more interesting results, 
see Ko's book~%
\cite{ko91:_comput_compl_of_real_funct}, 
survey~%
\cite{ko98:_polyn_time_comput_analy} 
or Weihrauch's book~%
\cite[Section~7.3]{weihrauch00:_comput_analy}. 

\subsubsection{Difficulties in generalizing to other spaces}

Unfortunately, 
this approach does not extend much further. 
For example, 
a naive extension to 
real functions on $\Rset$ (instead of $[0, 1]$) does not work: 
polynomial-time $(\rhosd, \rhosd)$-computability 
tends to fail for trivial reasons. 

\begin{example}
Addition on $\Rset$
(Example~\ref{example: addition}) 
is not polynomial-time $([\rhosd, \rhosd], \rhosd)$-computable. 
For suppose that a machine $([\rhosd, \rhosd], \rhosd)$-computed it 
within a polynomial time bound~$p$. 
In particular, the machine has to write the first symbol of the output 
in $t := p (1)$ steps or fewer. 
Note that this first symbol must be $1$ if the sum is greater than $1$, 
and $-1$ if the sum is less than $-1$. 
In particular, 
it must be $1$ if the two summands are $2 ^{t+100}$ and $-2 ^{t+50}$, 
and $-1$ if they are $2 ^{t+50}$ and $-2 ^{t+100}$. 
However, the machine cannot tell between these two cases, 
because it can read at most $t$ symbols of the input in time. 
\end{example}

The trouble seems to be that 
the time bound is independent of the input. 
Compare this with the addition of integers (written in binary)
by the usual Turing machine. 
It is in polynomial time, because 
a large summand would make the ``input size'' big 
and thereby give the machine more time. 
For the same thing to happen for addition of the real numbers, 
we would need to talk about the ``size'' of the input 
and a time bound ``polynomial in'' it, 
but we do not have the notion of size for infinite sequences. 
We could simply require, as 
Hoover~%
\cite{hoover90:_feasib_real_funct_and_arith_circuit}
and Takeuti~%
\cite{takeuti01:_effec_fixed_point_theor_over} 
did (see Section~\ref{subsection: applications, real numbers}), 
that the time to output the $i$th bit below the decimal point 
may depend polynomially in both $i$ and 
the logarithm of the absolute value of the input real number. 
This would have the same effect as 
our proposed extension of the computation model, 
in this specific case of $\Rset$---%
but our point is that we want a coherent theory of computation
that is applicable to other spaces by just switching representations. 
There are many objects other than $\Rset$ that 
we want to give representations to. 
The objects for which the infinite string model gives 
reasonable notions of complexity
are limited, 
compared to what we can do at the level of computability 
(see the discussions in 
Ko~\cite[pp.\,57--58]{ko91:_comput_compl_of_real_funct}, 
Weihrauch~\cite{weihrauch03:_comput_compl_comput_metric_spaces}, 
and 
Brattka et al.~\cite[p.\,484]{brattka08:_tutor_comput_analy}). 
Because of this limitation, 
the complexity of 
operators 
has been mostly formulated in non-uniform ways. 
We quote examples of such theorems below. 
We will reformulate them uniformly later 
(Theorems \ref{theorem: convex hull uniform} and \ref{theorem: ivp uniform}). 

\subsubsection{Non-uniform results}
\label{subsubsection: non-uniform results}

The first pair of results are
the positive and negative statements about
the operator of taking the convex hull $\OpCH (S)$ of 
a closed set $S \subseteq [0, 1] ^2$. 

Polynomial-time computability of 
a set $S \subseteq [0, 1] ^2$ 
is defined as follows
(see e.g.\ Braverman~%
\cite{braverman05:_compl_of_real_funct}
for a discussion), 
using the usual complexity class~$\classP$. 
We say that $S$ is \emph{polynomial-time computable} 
if there is a function
$\varphi \tcolon \varSigma ^* \tto \{0, 1\}$
in $\classP$ such that, 
for any $n \in \Nset$ and $u$, $v \in \Dset$, 
\begin{itemize}
\item
$\varphi (u, v, 0 ^n) = 1$ if $\dist ((\lsem u \rsem, \lsem v \rsem), S) < 2 ^{-n}$, and 
\item
$\varphi (u, v, 0 ^n) = 0$ if $\dist ((\lsem u \rsem, \lsem v \rsem), S) > 2 \cdot 2 ^{-n}$, 
\end{itemize}
where $
 \dist (p, S) 
:=
 \inf _{q \in S} \lVert p - q \rVert
$ denotes 
the Euclidean distance of point $p \in \Rset ^2$ from $S$
(Figure~\ref{figure: set computability}). 
\begin{figure}
\begin{center}
\includegraphics[scale=1.1]{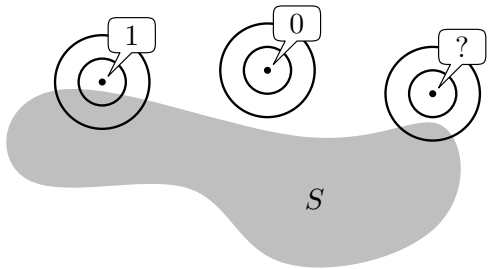}
\caption{Computing a set~$S$ means that, 
given $(u, v, 0 ^n)$, 
one can tell whether
the distance of the point $(\lsem u \rsem, \lsem v \rsem)$ from $S$ 
is less than $2 ^{-n}$ or more than $2 \cdot 2 ^{-n}$.
}
\label{figure: set computability}
\end{center}
\end{figure}
Likewise, $S$ is said to be 
\emph{nondeterministic polynomial-time computable}
if there is such a $\varphi$ in $\classNP$
(recall the asymmetry between the outcomes $1$ and $0$ 
in the definition of $\classNP$: 
we require an easily verifiable certificate for 
$(\lsem u \rsem, \lsem v \rsem)$ being close to $S$). 

Ko and Yu~%
\cite{ko08:_compl_of_convex_hulls_of}
and 
Zhao and M\"uller~%
\cite{zhao08:_compl_of_operat_compac_sets} 
essentially proved\footnote{\label{footnote: strong recognizability}%
Ko and Yu
state both the positive and the negative results 
(Theorems \ref{theorem: convex hull upper non-uniform} 
and \ref{theorem: convex hull lower non-uniform})
for polynomial-time \emph{strong recognizability} 
instead of our computability
\cite[Corollaries 4.3 and 4.6, respectively]{ko08:_compl_of_convex_hulls_of}, 
but their proof almost works for computability as well. 
For a discussion comparing the two notions, 
see Braverman~%
\cite{braverman05:_compl_of_real_funct}, 
where 
Ko's strong recognizability is called \emph{weak computability}. 
Zhao and M\"uller use the polynomial-time computability equivalent to ours 
and prove Theorem~%
\ref{theorem: convex hull upper non-uniform} 
\cite[Theorem~4.3]{zhao08:_compl_of_operat_compac_sets}. 
For the positive part, in fact they prove a uniform result
essentially equivalent to the positive part of 
our Theorem~\ref{theorem: convex hull uniform}
\cite[Theorem~4.1]{zhao08:_compl_of_operat_compac_sets}. 
Although they state the upper bound of ``exponential time'', 
their proof contains the argument that is 
necessary to derive the non-uniform $\classNP$ upper bound 
(our Theorem~\ref{theorem: convex hull lower non-uniform}) 
\cite[Lemma~4.2]{zhao08:_compl_of_operat_compac_sets}. 
}
the following non-uniform theorems 
about the complexity of taking the convex hull of a set. 

\begin{theorem}
\label{theorem: convex hull upper non-uniform}
If a closed set $S \subseteq [0, 1] ^2$ is 
polynomial-time computable, 
then its convex hull $\OpCH (S)$ is nondeterministic polynomial-time computable. 
\end{theorem}

\begin{theorem}
\label{theorem: convex hull lower non-uniform}
Unless $\classP = \classNP$, 
there exists a closed set~$S \subseteq [0, 1] ^2$
which is polynomial-time computable, 
but whose convex hull $\OpCH (S)$ is not. 
\end{theorem}

For $A \subseteq \Rset ^d$, 
let $\classC [A]$ be the set of continuous functions from $A$ to $\Rset$. 
The second pair of results concerns 
the initial value problem (IVP) of the differential equation
\begin{align}
\label{equation: ivp}
  h (0) & = 0, &
  h' (t) & = g \bigl( t, h (t) \bigr), 
\end{align}
where $g \in \classC [[0, 1] \times \Rset]$ is given 
and $h \in \classC [0, 1]$ is the unknown. 
It is well known (see 
\cite[beginning of Section~3]{kawamura10:_lipsc_contin_ordin_differ_equat}) 
that the solution~$h$ exists and is unique if $g$ is 
\emph{Lipschitz continuous} (in the second argument), 
that is, 
\begin{equation}
\label{equation: lipschitz}
 \lvert g (t, y _0) - g (t, y _1) \rvert \leq L \cdot \lvert y _0 - y _1 \rvert
\end{equation}
for some constant~$L$ independent of $t$, $y _0$, $y _1$. 
The following results state how complex $h$ can be, 
assuming that $g$ is polynomial-time computable. 
Since polynomial-time computability is defined only for functions with 
compact domain, 
we restrict $g$ to the rectangle $[0, 1] \times [-1, 1]$. 
If there is a solution $h \in \classC [0, 1]$ whose values stay in $[-1, 1]$ 
(in which case $h$ is unique, as mentioned above), 
we write $\OpLipIVP (g)$ for this $h$. 
Thus $\OpLipIVP$ is a partial function 
from $\classCL [[0, 1] \times [-1, 1]]$
to $\classC [0, 1]$, where 
the former set is the subset of 
$\classC [[0, 1] \times [-1, 1]]$ 
consisting of Lipschitz continuous functions. 

\begin{theorem}[{\cite[Section~4]{ko83:_comput_compl_of_ordin_differ_equat}%
\footnote{%
Ko~\cite{ko83:_comput_compl_of_ordin_differ_equat}
proved 
Theorems \ref{theorem: ivp upper non-uniform}
and \ref{theorem: ivp lower non-uniform}
with a condition slightly weaker than Lipschitz continuity. 
On the other hand, 
Ota et al.~\cite{ota12:_compl_of_smoot_ordin_differ_equat}
show that 
Theorem~\ref{theorem: ivp lower non-uniform} can be 
strengthened to yield a continuously differentiable function~$g$. 
}}]
\label{theorem: ivp upper non-uniform}
If $g \in \dom \OpLipIVP$ is 
polynomial-time computable, 
then $\OpLipIVP (g)$ is 
polynomial-space computable. 
\end{theorem}

\begin{theorem}[{\cite[Theorem~3.2]{kawamura10:_lipsc_contin_ordin_differ_equat}}]
\label{theorem: ivp lower non-uniform}
There is a polynomial-time computable function $
g \in \dom \OpLipIVP
$ such that $\OpLipIVP (g)$ is 
polynomial-space complete (in the sense defined in 
\cite{ko92:_comput_compl_of_integ_equat} or 
\cite{kawamura10:_lipsc_contin_ordin_differ_equat}). 
\end{theorem}

We can derive from Theorem~\ref{theorem: ivp lower non-uniform}
a statement of the form similar to 
Theorem~\ref{theorem: convex hull lower non-uniform}: 

\begin{corollary}[{\cite[Corollary~3.3]{kawamura10:_lipsc_contin_ordin_differ_equat}}]
\label{corollary: ivp lower non-uniform, weak}
Unless $\classP = \classPSPACE$, 
there is a real function~$g \in \dom \OpLipIVP$ which is polynomial-time computable
but $\OpLipIVP (g)$ is not. 
\end{corollary}

\section{Using functions as names}
\label{section: main}

We present the main definitions for our framework here. 

As we have noticed, 
the limitations of the approach 
with the infinite sequences in $\varSigma ^\Nset$ 
have to do with the fact that 
they do not carry enough information, 
and in particular their size is not defined. 
We replace $\varSigma ^\Nset$ with 
$\Bset$, a class of string functions 
which we will use as names of real numbers, sets and functions%
\footnote{Ko's formulation~%
\cite{ko91:_comput_compl_of_real_funct} 
already uses string functions instead of infinite strings, 
but it does not take full advantage of this extension, 
because 
queries to these functions are mostly restricted to unary strings~$0 ^n$. 
}. 

In Section~\ref{subsection: regular functions}, 
we introduce the class $\Bset$ of regular functions 
and define what it means for a machine to compute 
a multi-function from $\Bset$ to $\Bset$. 
Section~\ref{subsection: second-order polynomial} 
defines what it means for such a machine to 
run in polynomial time or space, 
thus introducing several complexity classes of multi-functions on $\Bset$. 
In Section~\ref{subsection: reduction and completeness}, 
we define reductions between such multi-functions, and 
discuss the resulting notions of hardness. 
Section~\ref{subsection: representations} 
extends this theory of computation on $\Bset$ 
to that on other sets $X$
by using representations of $X$ by $\Bset$. 

\subsection{Computation on regular functions}
\label{subsection: regular functions}

We say that a (total) function $\varphi \tcolon \varSigma ^* \tto \varSigma ^*$ is 
\emph{regular}%
\footnote{(Note added in May~2013) After the publication of this article in a journal, some authors started using ``length-monotone'', which is perhaps a more informative and better terminology.}
if it preserves relative lengths of strings in the sense that 
$\lvert \varphi (u) \rvert \leq \lvert \varphi (v) \rvert$ 
whenever $\lvert u \rvert \leq \lvert v \rvert$. 
We write $\Bset$ for the set of all regular functions. 
For the rest of this paper, 
we will be discussing the complexity of 
multi-functions from $\Bset$ to $\Bset$. 
The motivation for considering regular functions (rather than
all functions from $\varSigma ^*$ to $\varSigma ^*$) will be explained in 
Section~\ref{subsection: second-order polynomial}
where we define their lengths. 

Now we begin replacing the role of $\varSigma ^\Nset$ 
(Section~\ref{subsection: TTE computability})
by $\Bset$. 
This is a generalization, 
because an infinite string $a _0 a _1 \ldots \in \varSigma ^\Nset$ can be 
identified with
a regular function~$\varphi \in \Bset$ that 
\begin{enumerate}
\renewcommand{\theenumi}{\textup{(\alph{enumi})}}
\renewcommand{\labelenumi}{\theenumi}
\item 
\label{enumi: predicate}
takes values of length $1$, and
\item
\label{enumi: unary}
depends only on the length of the argument, 
\end{enumerate}
by setting $\varphi (0 ^n) = a _n$. 
In the following, 
observe that Definitions
\ref{definition: oracle machine}.\ref{enumi: oracle machine, deterministic} and 
\ref{definition: bounded by second-order polynomial} 
extend their counterparts in this sense. 

It sometimes makes sense to stop the generalization halfway, 
removing \ref{enumi: unary} only
and keeping \ref{enumi: predicate}. 
Let $\Pset \subseteq \Bset$ be the set of $\{0, 1\}$-valued regular functions. 

Instead of the machine 
that worked on infinite strings, 
we use an oracle Turing machine (henceforth just ``machine'') 
to convert regular functions to regular functions
(Figure~\ref{figure: oracle machine}): 

\begin{figure}
\begin{center}
\includegraphics[scale=1.1]{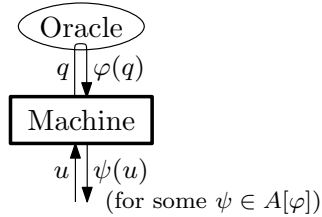}
\caption{A deterministic machine computing a multi-function $A$ from $\Bset$ to $\Bset$.}
\label{figure: oracle machine}
\end{center}
\end{figure}

\begin{definition}
\label{definition: oracle machine}
\begin{enumerate}
\item \label{enumi: oracle machine, deterministic}
A deterministic machine~$M$ 
\emph{computes} a multi-function $A$ from $\Bset$ to $\Bset$ 
if
for any $\varphi \in \dom A$, 
there is $\psi \in A [\varphi]$ such that 
$M$ on oracle~$\varphi$ and any string~$u$ outputs $\psi (u)$. 
\item
A nondeterministic machine~$M$ 
\emph{computes} a multi-function $A$ from $\Bset$ to $\Pset$ if 
for any $\varphi \in \dom A$, 
there is $\psi \in A [\varphi]$ such that 
$\psi (u) = 1$ if and only if 
$M$ on oracle~$\varphi$ and string~$u$ 
has at least one accepting computation path. 
\end{enumerate}
For the precise conventions for issuing and answering queries, 
follow any of 
\cite{mehlhorn76:_polyn_and_abstr_subrec_class, 
      kapron96:_new_charac_of_type_feasib, 
      ko91:_comput_compl_of_real_funct}.
\end{definition}

\subsection{Polynomial time and space}
\label{subsection: second-order polynomial}

Recall that regular functions are those that respect lengths
in the sense explained at the beginning of 
Section~\ref{subsection: regular functions}. 
In particular, they map strings of equal length
to strings of equal length. 
Therefore it makes sense to 
define the 
\emph{size}
$\lvert \varphi \rvert \tcolon \Nset \to \Nset$
of a regular function~$\varphi$ by 
$\lvert \varphi \rvert (\lvert u \rvert) = \lvert \varphi (u) \rvert$. 
It is a non-decreasing function from $\Nset$ to $\Nset$. 

Now we want to define what it means for a machine to run in polynomial time. 
Since $\lvert \varphi \rvert$ is a function, 
we begin by defining polynomials ``in'' a function, 
following the idea of Kapron and Cook~%
\cite{kapron96:_new_charac_of_type_feasib}. 
\emph{Second-order polynomials} 
(in type-$1$ variable~$\lit L$ and type-$0$ variable~$\lit n$) 
are defined 
inductively as follows: 
a positive integer is a second-order polynomial; 
the variable~$\lit n$ is also a second-order polynomial; 
if $P$ and $Q$ are second-order polynomials, 
then so are $P + Q$, $P \cdot Q$ and $\lit L (P)$. 
An example is 
\begin{equation}
 \label{equation: second-order polynomial example}
  \lit L \bigl( \lit L (\lit n \cdot \lit n) \bigr) + \lit L \bigl( \lit L (\lit n) \cdot \lit L (\lit n) \bigr) + \lit L (\lit n) + 4. 
\end{equation}
A second-order polynomial~$P$ specifies a 
function, which we also denote by $P$, 
that takes a function $L \tcolon \Nset \tto \Nset$
to another function $P (L) \tcolon \Nset \tto \Nset$
in the obvious way. 
For example, 
if $P$ is the above second-order polynomial~\eqref{equation: second-order polynomial example} 
and $L (x) = x ^2$, 
then $P (L)$ is given by
\begin{equation}
  P (L) (x) = \bigl( (x \cdot x) ^2 \bigr) ^2 + (x ^2 \cdot x ^2) ^2 + x ^2 + 4 
            = 2 \cdot x ^8 + x ^2 + 4. 
\end{equation}
As in this example, 
$P (L)$ is a (usual first-order) polynomial if $L$ is. 

\begin{definition}
 \label{definition: bounded by second-order polynomial}
A (deterministic or nondeterministic) machine~$M$
runs in (\emph{second-order}) \emph{polynomial time}
if there is a second-order polynomial~$P$ such that, 
given any $\varphi \in \Bset$ as oracle 
and any $u \in \varSigma ^*$ as input, 
$M$ halts within $P (\lvert \varphi \rvert) (\lvert u \rvert)$ steps
(regardless of the nondeterministic choices). 
Define \emph{polynomial space} analogously 
by counting the number of visited cells on all (input, work, output and query) tapes. 
\end{definition}

This extends 
Definition~\ref{definition: polynomial time stream machine}, 
because 
when $\varphi$ satisfies \ref{enumi: predicate} 
(of Section~\ref{subsection: regular functions}), 
the size $\lvert \varphi \rvert$ is constant, 
and the bound $P (\lvert \varphi \rvert) (\lvert u \rvert)$ 
reduces to a (first-order) polynomial in $\lvert u \rvert$. 

\begin{definition}
 \label{definition: functional classes}
\begin{enumerate}
\item
\label{enumi: functional classes functional}
We write $\classFPtwo$ (resp.\ $\classFPSPACEtwo$) 
for the class of 
multi-functions from $\Bset$ to $\Bset$ 
computed by a deterministic machine that 
runs in second-order polynomial time (resp.\ space). 
\item
We write $\classPtwo$ (resp.\ $\classNPtwo$) 
for the class of 
multi-functions from $\Bset$ to $\Pset$ 
computed by a deterministic (resp.\ nondeterministic)
machine~$M$ that runs in polynomial time. 
\end{enumerate}
\end{definition}

Note that unlike the type-one counterparts, 
it is easy to separate
$\classFPSPACEtwo$ from $\classFPtwo$
and $\classNPtwo$ from $\classPtwo$, 
because the former classes contain 
functions that 
depend on exponentially many values of the given oracle. 

It is also easy to see that 
these classes respect the corresponding type-one classes: 

\begin{lemma}
\label{lemma: maps computable}
\begin{enumerate}
\item \label{enumi: maps computable, functional}
Functions in $\classFPtwo$ (resp.\ $\classFPSPACEtwo$) 
map regular functions in $\classFP$ into $\classFP$
(resp.\ $\classFPSPACE$ into $\classFPSPACE$). 
\item \label{enumi: maps computable, predicate}
Functions in $\classPtwo$ (resp.\ $\classNPtwo$) 
map regular functions in $\classFP$ into $\classP$
(resp.\ $\classNP$). 
\end{enumerate}
\end{lemma}

\subsubsection*{Why we use only regular functions}

The idea of using second-order polynomials as a bound on time and space
comes from Kapron and Cook's characterization~%
\cite{kapron96:_new_charac_of_type_feasib}
of Mehlhorn's class~%
\cite{mehlhorn76:_polyn_and_abstr_subrec_class}
of \emph{polynomial-time computable operators}\footnote{%
  Kapron and Cook~%
\cite{kapron96:_new_charac_of_type_feasib}
  call them 
\emph{Basic Feasible Functionals}
  or 
\emph{Basic Polynomial-Time Functionals}. 
}. 
This is a class of (total) functionals $
F \tcolon (\varSigma ^* \tto \varSigma ^*) \times \varSigma ^* \tto \varSigma ^*
$, but they can be regarded as $
F \tcolon (\varSigma ^* \to \varSigma ^*) \to (\varSigma ^* \to \varSigma ^*)
$ by writing $F (\varphi) (x)$ for $F (\varphi, x)$. 
Kapron and Cook define the size of 
$\varphi \tcolon \varSigma ^* \tto \varSigma ^*$ by 
\begin{align}
\label{equation: size of functions}
 \lvert \varphi \rvert (n) & = \max _{\lvert u \rvert \leq n} \lvert \varphi (u) \rvert, 
&
 n \in \Nset. 
\end{align}
Note that our definition of size for regular $\varphi$ is a special case of this. 
They then defined the class of polynomial-time functionals
in the same way as Definition~\ref{definition: functional classes}.\ref{enumi: functional classes functional}. 
(We added $\classFPSPACEtwo$ by analogy.)

We have restricted attention to regular functions. 
This is because, 
in order to obtain reasonable complexity notions, 
it seems necessary for a machine to be able to tell 
when the time bound is reached. 
Note that for usual (type-one) computation, 
it was easy to find $\lvert x \rvert$ given $x$, 
and thus to clock the machine with the time bound $p (\lvert x \rvert)$ 
for a fixed polynomial~$p$. 
In contrast, 
finding the value \eqref{equation: size of functions} for a given $\varphi$ 
in general
requires exponentially many queries to $\varphi$. 
For regular $\varphi$, 
we can easily find $\lvert \varphi \rvert (n)$ for each $n$, 
and thus the second-order polynomial $P (\lvert \varphi \rvert) (\lvert u \rvert)$
is a bound ``time-constructible'' from $\varphi$ and $u$. 

Imposing regularity is hardly a restriction for our purpose, 
because our intention is to use these functions as names of 
real numbers, sets and functions, 
and we can simply require that 
valid names are those that have been ``padded'' to be regular. 
More precisely, 
there is an efficient machine 
that takes as oracles a possibly irregular function~$\varphi'$ and 
a regular function~$\psi$ dominating its length (i.e., 
$\lvert \varphi' (u) \rvert \leq \lvert \psi \rvert (\lvert u \rvert)$ 
for any string~$u$), 
and delivers a regular function $\varphi$ 
such that $\lvert \varphi \rvert = \lvert \psi \rvert$ 
and $\varphi (u) = \varphi' (u) \# \# \ldots \#$. 
Thus we use $\varphi$, instead of $\varphi'$, as the name. 
In many situations we can find such a dominating function~$\psi$. 

For later use
we define the pairing function for regular functions as follows
(we have been and will be using the tupling functions for strings, 
which we do not write explicitly): 
for $\varphi$, $\psi \in \Bset$, 
define 
$\langle \varphi, \psi \rangle \in \Bset$
by setting 
$\langle \varphi, \psi \rangle (0 u) = \varphi (u) 1 0 ^{\lvert \psi (u) \rvert}$ and 
$\langle \varphi, \psi \rangle (1 u) = \psi (u) 1 0 ^{\lvert \varphi (u) \rvert}$
(we pad the strings to make $\langle \varphi, \psi \rangle$ regular). 
Let $
 \langle \varphi, \psi, \theta \rangle 
=
 \langle \langle \varphi, \psi \rangle, \theta \rangle
$, etc. 

\subsection{Reduction and completeness}
\label{subsection: reduction and completeness}

Here we define reductions between multi-functions on $\Bset$ and 
discuss hardness with respect to these reductions. 

\subsubsection{Reductions}

Recall that the usual many-one reduction between
multi-functions $A$, $B$ from $\varSigma ^*$ to $\varSigma ^*$ is defined as follows: 
we say that $A$ many-one reduces to $B$ (written $A \redonem B$)
if there are (total) functions $r$, $t \in \classFP$ such that 
for any $u \in \dom A$, 
we have $r (u, v) \in A [u]$ whenever $v \in B [t (u)]$---that is, 
we have a function $t$ that converts an input for $A$ to an input for $B$, 
and another function $r$ that converts an output of $B$ to an output of $A$
(Figure~\ref{figure: type-one reductions}, left). 
\begin{figure}
\begin{center}\small
\hfill
\parbox[t]{100pt}{\centering\includegraphics[scale=1.1]{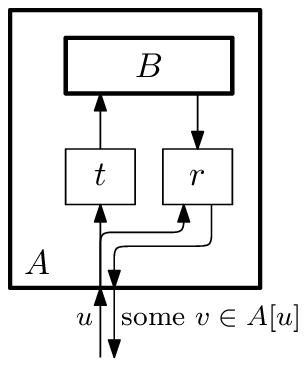}\\[3pt]$A \redonem B$}
\hfill
\parbox[t]{100pt}{\centering\includegraphics[scale=1.1]{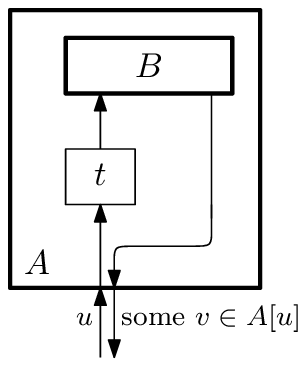}\\[3pt]$A \redonep B$}
\hfill
\parbox[t]{100pt}{\centering\includegraphics[scale=1.1]{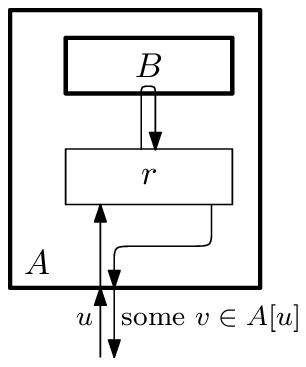}\\[3pt]$A \redoneT B$}
\hfill\mbox{}
\caption{Reductions between multi-functions $A$, $B$ on $\varSigma^*$.}
\label{figure: type-one reductions}
\end{center}
\end{figure}
The many-one reduction~$\mathord\redonep$ between predicates
is defined as the special case 
where we do not convert the output, i.e., 
$r (u, v) = v$ 
(Figure~\ref{figure: type-one reductions}, middle). 
Since multi-functions over $\Bset$ also get a function as input, 
the analogous definition of reductions 
involves one more converter~$s$: 

\begin{definition}
\label{definition: many-one reduction}
\begin{enumerate}
\item
Let $A$ and $B$ be multi-functions from $\Bset$ to $\Bset$. 
We say that $A$
\emph{many-one reduces} 
to $B$ 
(written $A \redtwom B$) if 
there are functions $r$, $s$, $t \in \classFPtwo$ such that 
for any $\varphi \in \dom A$, 
we have $s (\varphi) \in \dom B$ and, 
for any $\theta \in B [s (\varphi)]$, 
the function that maps each string~$x$ to 
$r (\varphi) (x, \theta (t (\varphi) (x)))$
is in $A [\varphi]$ 
(Figure~\ref{figure: intermediate many-one reduction}, 
\begin{figure}
\begin{center}\small
\hfill
\parbox[t]{150pt}{\centering\includegraphics[scale=1.1]{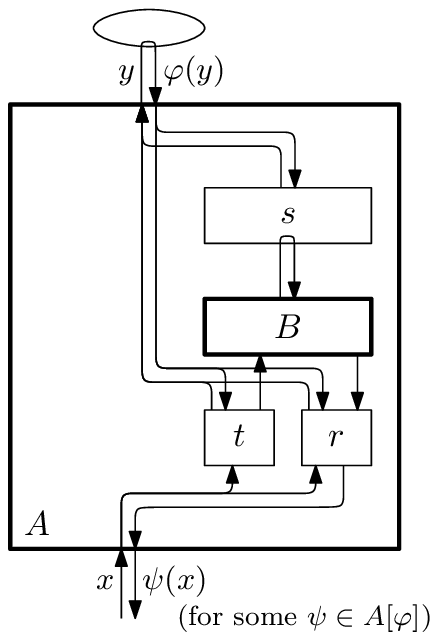}\\[3pt]$A \redtwom B$}
\hfill
\parbox[t]{150pt}{\centering\includegraphics[scale=1.1]{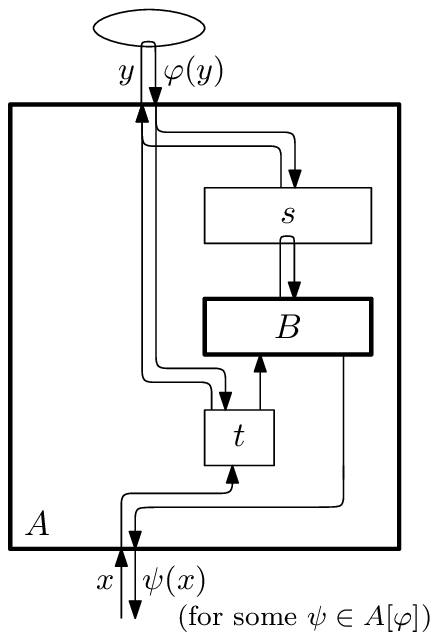}\\[3pt]$A \redtwop B$}
\hfill\mbox{}
\\[12pt]
\hfill
\parbox[t]{150pt}{\centering\includegraphics[scale=1.1]{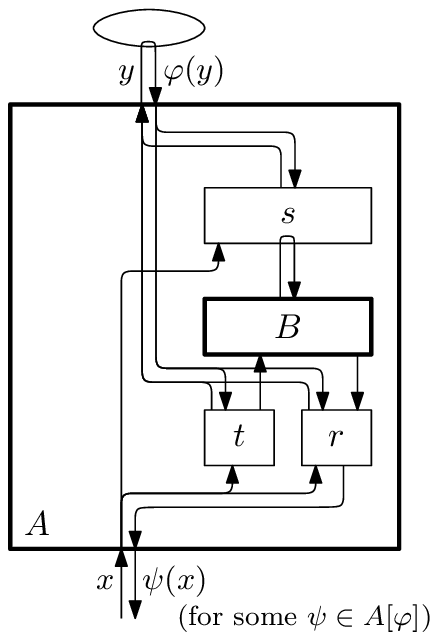}\\[3pt]$A$ many-one reduces to $B$\\in the sense of \cite{beame98:_relat_compl_of_np_searc_probl}}
\hfill
\parbox[t]{150pt}{\centering\includegraphics[scale=1.1]{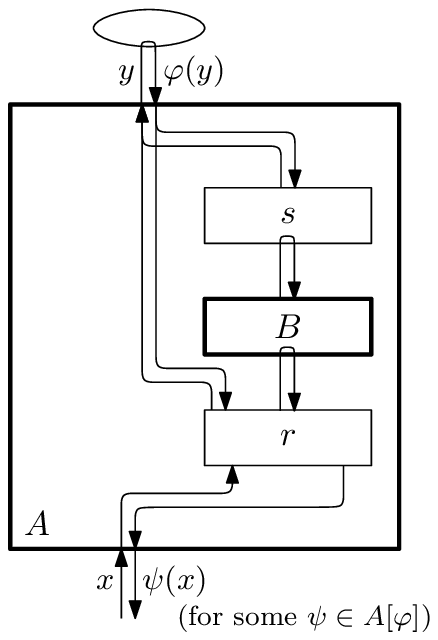}\\[3pt]$A \redtwoT B$}
\hfill\mbox{}
\caption{Reductions between multi-functions on $\Bset$.}
\label{figure: intermediate many-one reduction}
\end{center}
\end{figure}
top left). 
\item \label{enumi: many-one reduction, predicate}
Let $A$ and $B$ be multi-functions from $\Bset$ to $\Pset$. 
We write $A \redtwop B$ if 
there are functions $s$, $t \in \classFPtwo$ such that 
for any $\varphi \in \dom A$, we have $s (\varphi) \in \dom B$ and, 
for any $\theta \in B [s (\varphi)]$, the function
$\theta \circ t (\varphi)$ is in $A [\varphi]$ 
(Figure~\ref{figure: intermediate many-one reduction}, top right). 
\end{enumerate}
\end{definition}

The design of these reductions is somewhat arbitrary. 
We chose them simply because 
they are strong enough for our hardness results
(Theorems \ref{theorem: convex hull uniform} 
and \ref{theorem: ivp uniform}).
What 
Beame et al.~\cite{beame98:_relat_compl_of_np_searc_probl}
call the ``many-one reduction'' between type-two functions
is slightly stronger than ours in that 
it passes the string input $u$ 
not only to $t$ and $r$
but also to $s$
(Figure~\ref{figure: intermediate many-one reduction}, bottom left). 
See the comment after Lemma~\ref{lemma: maps complete} 
for the reason we did not choose this definition. 

Another reasonable notion of reduction is 
the one on the bottom right of Figure~\ref{figure: intermediate many-one reduction}: 

\addtocounter{theorem}{-1}
\begin{definition}[continued]
\begin{enumerate}
\addtocounter{enumi}{2}
\item \label{enumi: Weihrauch reduction}
Let $A$ and $B$ be multi-functions from $\Bset$ to $\Bset$ (or to $\Pset$). 
We say that $A$ \emph{Weihrauch reduces} to $B$ (written $A \redtwoT B$) if 
there are functions $r$, $s \in \classFPtwo$ such that 
for any $\varphi \in \dom A$, 
we have $r (\langle \varphi, \psi \rangle) \in A [\varphi]$ 
whenever $\psi \in B [s (\varphi)]$. 
\end{enumerate}
\end{definition}

This is a polynomial-time version of the continuous reduction used by 
Weihrauch~%
\cite{weihrauch92:_degrees_discon_some_transl_betwee}
to compare the degrees of discontinuity
of translators between real number representations
(see also Brattka and Gherardi~%
\cite{brattka11}). 
Note that, while this reduction is 
somewhat analogous to the type-one Turing reduction 
(Figure~\ref{figure: type-one reductions}, right), 
it also formally resembles the definition of $\redonem$. 
The many-one reduction $\redtwom$ is the special case of 
this reduction $\redtwoT$
where $r$ can query $\psi$ only once. 
Beame et al.~\cite{beame98:_relat_compl_of_np_searc_probl} 
define an even stronger ``Turing reduction''. 

In this paper, we will formulate our hardness results mainly using 
$\redtwom$ and $\redtwop$ 
because they give stronger results than $\redtwoT$ would. 
The advantage and disadvantage of this choice
will be discussed 
in Section~\ref{subsection: applications, real functions}
after Theorem~\ref{theorem: ivp uniform} 
and Corollary~\ref{corollary: ivp uniform, strong reduction}. 

\subsubsection{Hardness}

Now that we have the classes 
(Definition~\ref{definition: functional classes}) and 
reductions 
(Definition \ref{definition: many-one reduction}), 
we can talk about hardness. 
A multi-function~$B$ from $\Bset$ to $\Bset$ is 
\emph{$\classFPSPACEtwo$-$\redtwom$-hard} if 
$A \redtwom B$ for every $A \in \classFPSPACEtwo$.
It is said to be 
\emph{$\classFPSPACEtwo$-$\redtwom$-complete}
if moreover it is in $\classFPSPACEtwo$. 
We define $\classNPtwo$-$\redtwop$-hardness of 
multi-functions from $\Bset$ to $\Pset$ similarly
(note that $\classNPtwo$ is not closed under $\redtwom$). 

The following lemma states roughly that 
a $\classtwofont C$-$\leq ^2$-hard multi-function~$B$
maps some function $\psi \in \classFP \cap \Bset$ to a 
$\classfont C$-$\leq ^1$-hard function, 
where $\classfont C$ and $\leq ^1$ are the 
type-one versions of 
the class $\classtwofont C$ and the reduction $\leq ^2$. 
But since $B [\psi]$ may consist of more than one function, 
we need to assert hardness for 
the multi-function $\bigcup (B [\psi])$ 
defined as follows: 
for a nonempty set $F$ of (single-valued total) functions from $X$ to $Y$, 
we write 
$\bigcup F$ to mean the multi-function from $X$ to $Y$
defined by $(\bigcup F) [x] = \{\, f (x) : f \in F \,\}$. 
Saying that the multi-function $\bigcup F$ is hard is 
a stronger claim than saying that each of the functions in $F$ is hard, 
because the former requires that one reduction work for all functions in $F$. 
We need to state the following lemma in this stronger form 
in order to derive Lemma~\ref{lemma: maps complete represented} later. 

\begin{lemma}
\label{lemma: maps complete}
\begin{enumerate}
\item
\label{enumi: maps complete, functional}
Let $B$ be an $\classFPSPACEtwo$-$\redtwom$-complete 
multi-function from $\Bset$ to $\Bset$. 
Then there is $\psi \in \dom B \cap \classFP$ such that 
$\bigcup (B [\psi])$ is $\classFPSPACE$-$\redonem$-complete. 
\item
\label{enumi: maps complete, predicate}
Let $B$ be an $\classNPtwo$-$\redtwop$-complete 
multi-function from $\Bset$ to $\Pset$. 
Then there is $\psi \in \dom B \cap \classFP$ such that 
$\bigcup (B [\psi])$ is $\classNP$-$\redonep$-complete. 
\end{enumerate}
\end{lemma}

\begin{proof}
We only prove the first claim (the second claim is similar). 
There is a function $A \in \classFPSPACEtwo$ 
that maps some function $\varphi \in \classFP \cap \Bset$ to
an $\classFPSPACE$-$\redonem$-complete function. 
Since $A \redtwom B$, 
there are functions $r$, $s$, $t \in \classFPtwo$ 
as in Definition~\ref{definition: many-one reduction}. 
By Lemma~\ref{lemma: maps computable}, 
we have $r (\varphi)$, $s (\varphi)$, $t (\varphi) \in \classFP$. 
Let $\psi = s (\varphi)$. 
Since 
$r (\varphi)$ and $t (\varphi)$ give a reduction 
$A (\varphi) \redonem \bigcup (B [\psi])$, and 
$A (\varphi)$ is $\classFPSPACE$-$\redonem$-complete, 
$\bigcup (B [\psi])$ is also $\classFPSPACE$-$\redonem$-complete. 
\end{proof}

We note that Lemma~\ref{lemma: maps complete}
would not have been true, 
if in the definition of reductions 
we had fed $s$ with the string input 
as Beame et al.~\cite{beame98:_relat_compl_of_np_searc_probl} do
(see the comment after Definition~\ref{definition: many-one reduction}). 
For let $\leq ^2 _*$ be the reduction 
which is like $\leq ^2$ but 
feeds $s$ with the string input, 
and let $B$ be a $\classtwofont C$-$\leq ^2$-complete multi-function. 
Then the multi-function~$B'$ defined by 
\newcommand{\const}{\mathrm{const}}
\begin{align}
  \dom B' 
& 
 = 
  \{\, 
   \langle \const _u, \varphi \rangle 
  :
   u \in \varSigma ^*, \ \varphi \in \dom B 
  \,\}, 
\\
  B' [\langle \const _u, \varphi \rangle] 
& 
 =
  \{\, \const _{\psi (u)} : \psi \in B [\varphi] \,\}, 
\end{align}
where 
$\const _u \in \Bset$ denotes the constant function with value~$u$, 
is $\classtwofont C$-$\leq ^2 _*$-complete 
by the modified reduction where 
``$s$ does the job that $t$ used to do''. 
Yet each one of the values of $B'$ is a constant function. 

\subsubsection{Some complete problems}

Let $\classPSPACEtwo$ be the subclass of $\classFPSPACEtwo$ consisting of 
multi-functions from $\Bset$ to $\Pset$.
Here we list some 
$\classNPtwo$- and $\classPSPACEtwo$-$\redtwop$-complete problems. 
Their completeness can be proved by relativizing 
the well-known $\classNP$- and $\classPSPACE$-$\redonep$-completeness
in a straightforward way. 

We begin with $\classNPtwo$-$\redtwop$-complete problems. 
For a non-decreasing function $\mu \tcolon \Nset \tto \Nset$, 
define $\overline \mu \in \Bset$ by 
$\overline \mu (u) = 0 ^{\mu (\lvert u \rvert)}$. 
A \emph{Boolean formula involving predicate symbol $\lit p$}
is an expression built up inductively from 
Boolean variables $a _1$, $a _2$, \ldots 
using the connectives
$f _1 \land f _2$, $f _1 \lor f _2$, $\lnot f _1$ and 
$\lit p (f _1, \ldots, f _n)$ (the arity~$n$ can vary) for 
any previously obtained formulas $f _1$, $f _2$, \ldots. 

\begin{lemma}
\label{lemma: nptwo-copmlete}
The following partial functions 
$\textsc{ntime} ^2$, $\textsc{exist} ^2$ and $\textsc{sat} ^2$
from $\Bset$ to $\Pset$
are $\classNPtwo$-$\redtwop$-complete: 
\begin{itemize}
\item 
$\dom \textsc{ntime} ^2$ consists of 
all triples $\langle M, \overline \mu, \varphi \rangle$ such that 
$M$ is a (program of a) nondeterministic (oracle Turing) machine, 
$\mu \tcolon \Nset \tto \Nset$ is non-decreasing, 
$\varphi \in \Bset$, 
and for any string~$u$, 
all computation paths of $M ^\varphi (u)$ 
halt in time $\mu (\lvert u \rvert)$
(this $M$ is a string, so we encode it as the constant function 
taking this string as value). 
For any such triple and a string $u$, 
we have $
 \textsc{ntime} ^2 (\langle M, \overline \mu, \varphi \rangle) (u) 
=
 1
$ if and only if $M ^\varphi (u)$ has an accepting path. 
\item
$\dom \textsc{exist} ^2 = \Pset$. 
For any $p \in \Pset$, $u \in \varSigma ^*$ and $n \in \Nset$, 
we have $
 \textsc{exist} ^2 (p) (u, \allowbreak 0 ^n) 
=
 1
$ if and only if there is a string $v$ of length~$n$ with $
 p (u, v) 
= 
 1
$. 
\item
$\dom \textsc{sat} ^2 = \Pset$. 
For any $p \in \Pset$ and any string~$u$, 
we have $\textsc{sat} ^2 (p) (u) = 1$ if and only if 
$u$ is a Boolean formula involving a predicate symbol~$\lit p$
and it is made satisfiable when $\lit p$ is interpreted as $p$. 
\end{itemize}
\end{lemma}





If $\varphi$ is a Boolean formula involving predicate symbol $\lit p$, 
then an expression of the form 
\begin{equation}
 Q _1 a _1 \ldotp Q _2 a _2 \ldotp Q _3 a _3 \ldots Q _k a _k \ldotp
 \varphi (a _1, \ldots, a _k), 
\end{equation}
where each $Q _i$ is either $\forall$ or $\exists$, 
is called a \emph{quantified Boolean formula involving predicate symbol~$\lit p$}. 
Its truth value is determined in the obvious way 
relative to the predicate to be substituted into $\lit p$.  

\begin{lemma}
\label{lemma: pspacetwo-complete}
The following partial functions 
$\textsc{space} ^2$, $\textsc{power} ^2$, $\textsc{qbf} ^2$
from $\Bset$ to $\Pset$
are $\classPSPACEtwo$-$\redtwop$-complete: 
\begin{itemize}
\item
$\dom \textsc{space} ^2$ consists of 
all triples $\langle M, \overline \mu, \varphi \rangle$ such that 
$M$ is a (program of a) deterministic (oracle Turing) machine, 
$\mu \tcolon \Nset \tto \Nset$ is non-decreasing, 
$\varphi \in \Bset$, 
and for any string~$u$, 
the computation $M ^\varphi (u)$
uses no more than $\mu (\lvert u \rvert)$ tape cells 
and either accepts or rejects 
(this $M$ is a string, so we encode it as the constant function 
taking this string as value). 
For any such triple and a string $u$, 
we have $
 \textsc{space} ^2 (\langle M, \overline \mu, \varphi \rangle) (u) 
=
 1
$ if and only if $M ^\varphi (u)$ accepts. 
\item
$\dom \textsc{power} ^2$ consists of all 
$f \in \Bset$ that are length-preserving
(i.e., $\lvert f \rvert = \id$). 
For any such $f$ and a string $u$, 
we have $
 \textsc{power} ^2 (f) (u) 
=
 1
$ if and only if $
 f ^{2 ^{\lvert u \rvert}} (u) 
= 
 0 ^{\lvert u \rvert}
$, 
where 
we write $f ^k$ for $f$ iterated $k$ times. 
\item 
$\dom \textsc{qbf} ^2 = \Pset$. 
For any $p \in \Pset$ and any string~$u$, 
we have $\textsc{qbf} ^2 (p)(u) = 1$ if and only if 
$u$ is a quantified Boolean formula involving a predicate symbol
and it is made true by $p$. 
\end{itemize}
\end{lemma}

\subsection{Representations}
\label{subsection: representations}

As we replaced $\varSigma ^\Nset$ by $\Bset$, 
we extend the notion of representations accordingly: 
a 
\emph{representation}
$\gamma$ of a set~$X$ is 
a surjective partial function from $\Bset$ to $X$. 
Computation relative to representations 
is again formulated by Definition~\ref{definition: realization}. 
This defines what it means for a multi-function $F$ from $X$ to $Y$, 
where $X$ and $Y$ are sets equipped with representations $\gamma$ and $\delta$, 
respectively, 
to be in $(\gamma, \delta)$-$\classtwofont C$, 
where $\classtwofont C$ is 
one of the classes we have defined, such as $\classFPtwo$ and $\classFPSPACEtwo$. 
This $\classtwofont C$ can be $\classPtwo$ or $\classNPtwo$
if $\dom \delta \subseteq \Pset$. 
Also, 
we say that $F$ is $(\gamma, \delta)$-$\classtwofont C$-$\leq$-\emph{hard/complete}
(for $\classtwofont C = \classFPSPACEtwo$, $\classNPtwo$)
if $\delta ^{-1} \circ F \circ \gamma$ 
(see Definition~\ref{definition: realization}) is $\classtwofont C$-$\leq$-hard/complete. 

\subsubsection{Translation and equivalence}
\label{subsubsection: translation and equivalence}

\newcommand{\translT}{\leq}
\newcommand{\equivT}{\equiv}

Here we discuss how 
the class $(\gamma, \delta)$-$\classtwofont C$ and 
$(\gamma, \delta)$-$\classtwofont C$-$\leq$-hardness depend on 
the choice of representations $\gamma$ and $\delta$. 
For two representations $\delta$ and $\delta'$ of the same set, 
we write $\delta \translT \delta'$ if 
there is a function $F \in \classFPtwo$ 
that \emph{translates} $\delta$ to $\delta'$ in the sense that 
for all $\varphi \in \dom \delta$, 
we have $F (\varphi) \in \dom \delta'$ and $\delta (\varphi) = \delta' (F (\varphi))$. 
Thus $\delta$ is ``more informative'' or ``less generic'' than $\delta'$. 
It is easy to see the following. 

\begin{lemma}
\label{lemma: translation and computability}
Let $\classtwofont C$ be either $\classFPtwo$ or $\classFPSPACEtwo$. 
Let $\gamma$ and $\gamma'$ be representations of a set $A$, 
and $\delta$ and $\delta'$ be representations of a set $B$. 
If $\gamma' \translT \gamma$ and $\delta \translT \delta'$, 
then $
 (\gamma, \delta) \text- \classtwofont C 
\subseteq
 (\gamma', \delta') \text- \classtwofont C
$. 
\end{lemma}

We write $\gamma \equivT \gamma'$
if $\gamma \translT \gamma'$ and $\gamma' \translT \gamma$. 
Lemma~\ref{lemma: translation and computability} implies that 
the class $(\gamma, \delta)$-$\classFPtwo$ is 
invariant under replacing $\gamma$ or $\delta$
with $\equivT$-equivalent representations. 

By reversing the directions of the translations 
between $\gamma$, $\gamma'$ and $\delta$, $\delta'$ in the assumption, 
we get the implication between hardness results under different representations:

\begin{lemma}
\label{lemma: translation and completeness}
Let $\gamma$ and $\gamma'$ be representations of a set $A$, 
and $\delta$ and $\delta'$ be representations of a set $B$. 
If $\gamma \translT \gamma'$ and $\delta' \translT \delta$, 
then a $(\gamma, \delta)$-$\classFPSPACEtwo$-$\redtwoT$-hard multi-function
is $(\gamma', \delta')$-$\classFPSPACEtwo$-$\redtwoT$-hard. 
\end{lemma}

Here, 
$\redtwoT$ is the stronger reduction in 
Definition~\ref{definition: many-one reduction}.\ref{enumi: Weihrauch reduction}
(note that $\redtwom$ would not work). 

\subsubsection{Uniform and non-uniform statements}
\label{subsubsection: uniform non-uniform}

We say that an element $x \in X$ is in 
\emph{$\gamma$-$\classfont C$} 
(where $\classfont C$ is a usual complexity class of string functions, 
such as $\classFP$ and $\classFPSPACE$)
if it has a $\gamma$-name in $\classfont C$. 
It is 
\emph{$\gamma$-$\classfont C$-$\leq$-complete}
(where $\leq$ is either $\redonem$, $\redonep$ or $\redoneT$)
if $\bigcup (\gamma ^{-1} [x])$ 
(where $\mathord\cdot ^{-1}$ is the inverse image, 
and $\bigcup$ is defined before Lemma~\ref{lemma: maps complete}) 
is $\classfont C$-$\leq$-complete. 
Lemmas \ref{lemma: maps computable} and \ref{lemma: maps complete} yield
the following. 

\begin{lemma}
\label{lemma: maps computable represented}
Let $\gamma$ and $\delta$ be representations of sets $X$ and $Y$, 
respectively. 
\begin{enumerate}
\item 
\label{enumi: maps computable represented, functional}
A partial function $F \in (\gamma, \delta) \text- \classFPtwo$ 
maps elements of $\gamma \text- \classFP \cap \dom F$ into 
$\delta \text- \classFP$. 
Similarly for $\classFPSPACEtwo$ and $\classFPSPACE$
replacing $\classFPtwo$ and $\classFP$. 
\item 
\label{enumi: maps computable represented, predicate}
Suppose that $\dom \delta \subseteq \Pset$. 
Then a partial function $F \in (\gamma, \delta) \text- \classPtwo$ 
maps elements of $\gamma \text- \classFP \cap \dom F$ into 
$\delta \text- \classP$. 
Similarly for $\classNPtwo$ and $\classNP$
replacing $\classPtwo$ and $\classP$. 
\end{enumerate}
\end{lemma}

\begin{lemma}
\label{lemma: maps complete represented}
Let $\gamma$ and $\delta$ be representations of sets $X$ and $Y$, 
respectively. 
\begin{enumerate}
\item
\label{enumi: maps complete represented, functional}
A $(\gamma, \delta)$-$\classFPSPACEtwo$-$\redtwom$-complete 
partial function~$F$ maps
some element of $\gamma \text- \classFP \cap \dom F$ to 
a $\delta$-$\classFPSPACE$-$\redonem$-complete element of $Y$. 
\item
\label{enumi: maps complete represented, predicate}
Suppose that $\dom \delta \subseteq \Pset$. 
Then a $(\gamma, \delta)$-$\classNPtwo$-$\redtwop$-complete 
partial function~$F$ maps 
some element of $\gamma \text- \classFP \cap \dom F$ to 
a $\delta$-$\classNP$-$\redonep$-complete element of $Y$. 
\end{enumerate}
\end{lemma}

These lemmas will be used in Sections 
\ref{subsection: applications, real sets} and 
\ref{subsection: applications, real functions} 
to derive the non-uniform theorems from their uniform counterparts. 

\section{Applications}
\label{section: applications}

As noted in Section~\ref{subsection: regular functions}, 
our formulation can be viewed as a 
generalization of TTE 
achieved by removing the restrictions
\ref{enumi: predicate} and \ref{enumi: unary} 
on the oracle used as names. 
In the following three subsections, 
we will apply our theory to 
real numbers, real sets and real functions 
through representations $\rhoreal$, $\psiclosed$, and $\deltabox$, 
which exploit the removal of 
\ref{enumi: predicate}, 
\ref{enumi: unary}, 
and both, 
respectively. 

For representations $\gamma _0$ and $\gamma _1$ 
of $X _0$ and $X _1$, respectively, 
we can define the representation $[\gamma _0, \gamma _1]$
of the Cartesian product $X _0 \times X _1$ by $
  [\gamma _0, \gamma _1] (\langle \varphi _0, \varphi _1 \rangle) 
= 
  (\gamma _0 (\varphi _0), \gamma _1 (\varphi _1))
$. 

\subsection{Computation on real numbers}
\label{subsection: applications, real numbers}

Recall the representation~$\rhoreal$ of $\Rset$ by infinite sequences
(Section~\ref{subsection: TTE computability})
where a $\rhoreal$-name of a real number~$x$
was a list $u _0 \# u _1 \# \ldots$ 
of rational numbers~$\lsem u _i \rsem$ 
with $\lvert \lsem u _i \rsem - x \rvert < 2 ^{-i}$. 
We adopt this in a straightforward way
to define a representation~$\rhoreal$ 
that encodes real numbers by regular functions 
(we keep writing $\rhoreal$ by abuse of notation):  
we say that $\varphi \in \Bset$ is a $\rhoreal$-name of $x \in \Rset$ 
if $\varphi (0 ^i) \in \Dset$ and 
$\lvert \lsem \varphi (0 ^i) \rsem - x \rvert < 2 ^{-i}$ for each $i \in \Nset$. 
Thus we encode the same list in the values $\varphi (0 ^i)$. 

We write $\rhoreal | ^{[0, 1]}$ for 
the restriction of $\rhoreal$ to (names of) real numbers in the interval $[0, 1]$. 
It is easy to see
that the class $(\rhoreal | ^{[0, 1]}, \rhoreal)$-$\classFPtwo$ 
coincides with the polynomial-time computability 
that was defined in Section~\ref{subsubsection: TTE complexity compact}
using the signed digit representation~$\rhosd$. 
Recall that in the definition of $\rhosd$, 
we needed to forbid redundancy carefully. 
Now we do not have to worry too much about defining concise representations. 

Moreover, we obtain a reasonable notion of 
polynomial-time computability of real functions~$f$ 
on $\Rset$ (not just $[0, 1]$) 
without additional work: 
$(\rhoreal, \rhoreal)$-$\classFPtwo$
turns out to be a reasonable class
that coincides with the one by Hoover~%
\cite{hoover90:_feasib_real_funct_and_arith_circuit}, 
who required that 
the $2 ^{-m}$-approximation of the value~$f (t)$ 
should be delivered within time polynomial in $m$ and $\log \lvert t \rvert$ 
(this equivalence has been essentially observed by 
Lambov~%
\cite{lambov06:_basic_feasib_funct_in_comput_analy}).
Another equivalent definition appears in Takeuti~%
\cite{takeuti01:_effec_fixed_point_theor_over}, 
inspired by Pour-El's approach to computable analysis. 

\begin{example}
It is easy to see that 
binary addition and multiplication on $\Rset$ are in 
$([\rhoreal, \rhoreal], \allowbreak \rhoreal)$-$\classFPtwo$ 
by the algorithms suggested by 
Examples \ref{example: addition} and \ref{example: multiplication}. 
\end{example}

\begin{example}
The exponential function $\exp \tcolon \Rset \tto \Rset$
restricted to $[0, 1]$ is in $(\rhoreal | ^{[0, 1]}, \rhoreal)$-$\classFPtwo$, 
because $\exp t$ can be written as the sum of a series 
which is known to converge fast on $[0, 1]$
(that is, given a desired precision, the machine can tell 
how many initial terms it needs to compute). 
However, $\exp$ on the whole real line $\Rset$ is not 
in $(\rhoreal, \rhoreal)$-$\classFPtwo$, 
because it grows too fast. 
\end{example}

\begin{example}
 \label{example: sine}
The sine function $\sin \tcolon \Rset \tto \Rset$
is in $(\rhoreal, \rhoreal)$-$\classFPtwo$. 
To see this, 
note that just like $\exp$ in the previous example, 
$\sin$ is polynomial-time computable if 
restricted to, say, $[-4, 4]$. 
It is also possible, 
given $t \in \Rset$ and a desired precision, 
to find efficiently 
a number in $[-4, 4]$ which is 
close enough to $t$ modulo $2 \pi$. 
We can compute $\sin t$
by combining these algorithms. 
\end{example}

\begin{example}
A function can belong to $(\rhoreal, \rhoreal)$-$\classFPtwo$ 
without even an explicit description known. 
The 
\emph{trisector curves} 
(Figure~\ref{figure: trisector})
\begin{figure}
\begin{center}
\includegraphics[scale=0.3]{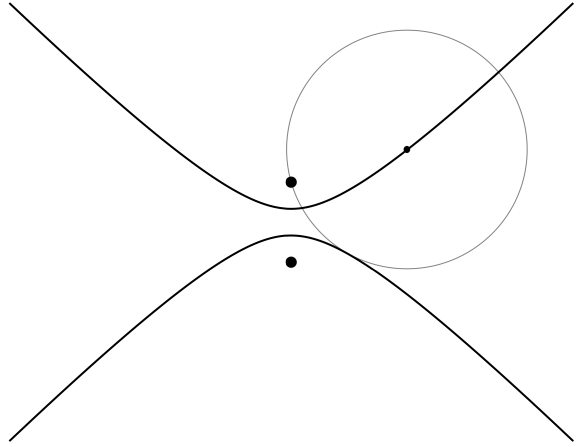}
\caption{The trisector curves between two points.}
\label{figure: trisector}
\end{center}
\end{figure}
between the points $(0, 1)$ and $(0, -1)$ in the plane are
the pair of sets $C _1$, $C _2 \subseteq \Rset ^2$ such that 
$C _1$ is the set of points equidistant from $(0, 1)$ and $C _2$, and 
$C _2$ is the set of points equidistant from $(0, -1)$ and $C _1$. 
Asano, Matou\v sek and Tokuyama~%
\cite{asano07:_distan_trisec_curve}
proved that such a pair $(C _1, C _2)$ exists and is unique
(see 
\cite{imai10:_distan_k_sector_exist} and 
\cite{kawamura12:_zone_diagr_in_euclid_spaces}
for simpler and more general proofs). 
They also showed that $C _1$ (as well as $C _2$, which is its mirror image) 
is a graph of a function $f \tcolon \Rset \to \Rset$ which is, 
in our terminology, 
in $(\rhoreal, \rhoreal)$-$\classFPtwo$. 
They conjecture that these curves are
different from any curve that was previously known. 
\end{example}

\subsection{Computation on real sets}
\label{subsection: applications, real sets}

\subsubsection{Representation of real sets}
Let $\mathcal A$ be the set of closed subsets of $[0, 1] ^2$. 
Define the representation~$\psiclosed$ of $\mathcal A$ as follows: 
let $\varphi \in \Pset$ be a $\psiclosed$-name of $S \in \mathcal A$ if 
it satisfies the two itemized conditions in 
Section~\ref{subsubsection: non-uniform results}. 
Note that this representation is more succinct than 
the one that we would be able to define using infinite sequences~%
\cite[Example~6.9]{weihrauch03:_comput_compl_comput_metric_spaces}.

Since $\dom \psiclosed \subseteq \Pset$, 
it makes sense to talk about 
$\psiclosed$-$\classNP$ and 
$(\psiclosed, \psiclosed)$-$\classNPtwo$ 
(Section \ref{subsubsection: uniform non-uniform}).
The following is immediate from the definition
of polynomial-time computability in Section~\ref{subsubsection: non-uniform results}. 

\begin{lemma}
\label{lemma: set polytime iff polytime name}
A set in $\mathcal A$ is (nondeterministic) polynomial-time computable 
if and only if it is in $\psiclosed$-$\classP$ 
($\psiclosed$-$\classNP$). 
\end{lemma}

\subsubsection{Complexity of the convex hull operator}
The operator~$\OpCH$ taking convex hulls 
(Section \ref{subsubsection: non-uniform results}) 
is a function from $\mathcal A$ to $\mathcal A$. 
We can state and prove the following uniform version of 
Theorems \ref{theorem: convex hull upper non-uniform}
and \ref{theorem: convex hull lower non-uniform}. 
As corollaries to this, 
we get 
Theorem~\ref{theorem: convex hull upper non-uniform}
by 
Lemmas \ref{lemma: maps computable represented}.%
\ref{enumi: maps computable represented, predicate} and 
\ref{lemma: set polytime iff polytime name}, 
and 
Theorem~\ref{theorem: convex hull lower non-uniform} by 
Lemmas \ref{lemma: maps complete represented}.%
\ref{enumi: maps complete represented, predicate}
and \ref{lemma: set polytime iff polytime name}. 

\begin{theorem}
\label{theorem: convex hull uniform}
$\OpCH$ is $(\psiclosed, \psiclosed)$-$\classNPtwo$-$\redtwop$-complete. 
\end{theorem}

\begin{proof}
The main technical ideas are 
already in Ko and Yu's proof of 
the non-uniform versions
(Theorems \ref{theorem: convex hull upper non-uniform} and 
\ref{theorem: convex hull lower non-uniform}), 
so we will only sketch the proof. 

That $\OpCH$ belongs to $(\psiclosed, \psiclosed)$-$\classNPtwo$ is no surprise: 
A point~$p$ belongs to $\OpCH (S)$ if there are two points $p'$ and $p''$ in $S$
such that $p$ is on the line segment $p' p''$. 
All we have to do is to guess $p'$ and $p''$ nondeterministically, 
with appropriate consideration of precision. 

For hardness, we need to modify the proof slightly, 
because, as we noted earlier, 
Ko and Yu's original results were about a weaker notion of computability: 
our computability of sets demands more in the sense that 
on query $(u, v, 0 ^n)$, where $u$, $v \in \Dset$ and $n \in \Nset$, 
if $(\lsem u \rsem, \lsem v \rsem)$ is within distance $2 ^{-n}$ from the set, 
then we must say $1$ 
(see the definition before 
Theorems \ref{theorem: convex hull upper non-uniform} 
and \ref{theorem: convex hull lower non-uniform}), 
whereas for weak computability both $0$ or $1$ are allowed in this case. 

We assume that the reader has 
Ko and Yu's proof~\cite[Corollary 4.6]{ko08:_compl_of_convex_hulls_of}
at hand. 
The proof of their Lemma~4.4 
begins by taking an arbitrary set $A \in \classNP$
and noting that there are $B \in \classP$ and a polynomial~$p$
such that $w \in A$ if and only if $(w, u) \in B$
for some string~$u$ of length exactly $p (\lvert w \rvert)$. 
Relativizing this, we take $A \in \classNPtwo$, and 
note that there are $B \in \classPtwo$ and a second-order polynomial~$P$
such that $A (\varphi) (w) = 1$ if and only if 
$B (\varphi) (w, u) = 1$ for 
some string~$u$ of length $P (\lvert \varphi \rvert) (\lvert w \rvert)$. 

We need to provide $s$ and $t$ of 
Definition~\ref{definition: many-one reduction}
which reduce $A$ to $\psiclosed ^{-1} \circ \OpCH \circ \psiclosed$.
We define $s$ by describing the set $S = \psiclosed (s (\varphi))$ 
for a given $\varphi \in \Bset$. 
The construction is 
similar to the original proof, 
replacing $p (n)$ by $P (\lvert \varphi \rvert) (n)$ 
and $B$ by $B (\varphi)$. 

\begin{figure}
\begin{center}
\includegraphics[scale=1.05]{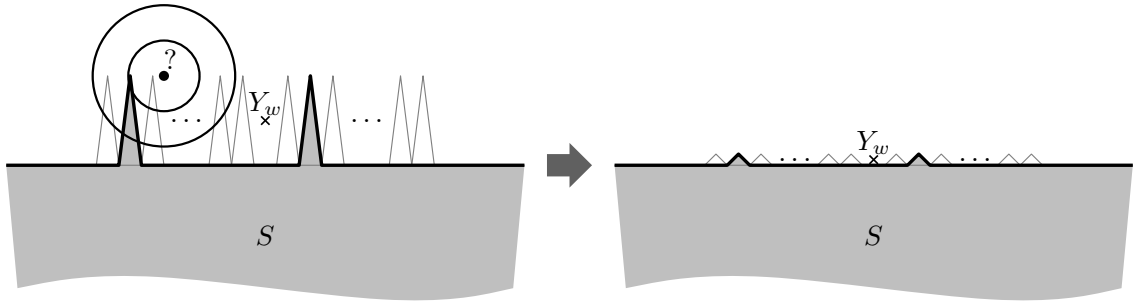}
\caption{Widget for reducing $\classNPtwo$ to $\OpCH$.
         We have $Y _w \in \OpCH (S)$ if and only if 
         there is $u$ such that 
         the slot for $(w, u)$ has a bump.
         In Ko and Yu's construction of $S$, 
         the bumps can be high (left), 
         and there can be a query that 
         requires the knowledge of $B (w, u)$ for many $u$. 
         We make the bumps low (right) in order to make $S$ 
         polynomial-time computable in our sense.}
\label{figure: convex hull widget}
\end{center}
\end{figure}
The original proof constructs, for each string~$w$, 
the widget in Figure~\ref{figure: convex hull widget} left
(or Figure~2 of \cite{ko08:_compl_of_convex_hulls_of}). 
In each of the left and right halves, 
there are exponentially many slots, one for each $u$, 
that have a bump if and only if $(w, u)$ is in $B$ (or $B (\varphi)$ for us). 
The point of this construction is that, 
while the set~$S$ is easy to compute, 
$\OpCH (S)$ is hard in the sense that
we can tell if $w$ is in $A$ (or $A (\varphi)$) by 
checking whether the middle point $Y _w$ belongs to $\OpCH (S)$. 
But this
$S$ is not easy in our sense, 
because in order to answer the query shown in 
Figure~\ref{figure: convex hull widget}, 
we need to know $B (\varphi) (w, u)$ 
for exponentially many $u$. 
To avoid this, 
we make the bumps low, 
so they make an angle of at most 
$45 ^\circ$ (Figure~\ref{figure: convex hull widget} right). 
This ensures that any one query to the $\psiclosed$-name of $S$
can be answered by checking $B (\varphi) (w, u)$ for 
at most one $(w, u)$, 
making $s$ computable in polynomial time. 

The function~$t$ queries 
whether the point $Y _w$ is in $\OpCH (S)$ 
with appropriate precision. 
Note that $t$ needs access to $\varphi$ in order to 
find the location of $Y _w$ and the necessary precision. 
\end{proof}

\subsection{Computation on real functions}
\label{subsection: applications, real functions}

\subsubsection{Representation of real functions}
We say that a non-decreasing function $\mu \tcolon \Nset \tto \Nset$ is a 
\emph{modulus of continuity}
of a function $f \in \classC [0, 1]$ if 
for all $n \in \Nset$ and $t$, $t' \in [0, 1]$ such that 
$\lvert t - t' \rvert \leq 2 ^{-\mu (n)}$, we have 
$\lvert f (t) - f (t') \rvert \leq 2 ^{-n}$ (Figure~\ref{figure: moc}). 
Note that any $f \in \classC [0, 1]$ is uniformly continuous
and thus has a modulus of continuity. 
\begin{figure}
\begin{center}
\includegraphics[scale=1.1]{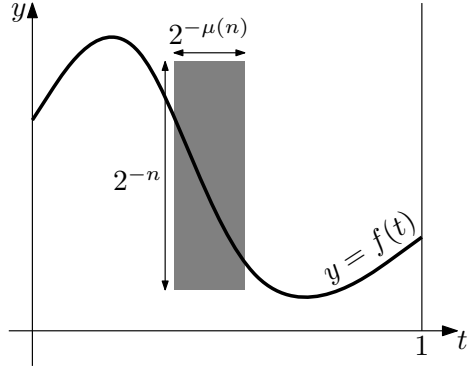}
\caption{Modulus of continuity~$\mu$ of a real function $f \in \classC [0, 1]$.}
\label{figure: moc}
\end{center}
\end{figure}

Define the representation $\deltabox$ of $\classC [0, 1]$ as follows
(see Lemmas 
\ref{lemma: box polytime iff polytime name} and 
\ref{lemma: universality} below 
for the reasons why we believe $\deltabox$
to be a natural representation): 
for $\mu \tcolon \Nset \tto \Nset$ and $\varphi \in \Bset$, 
we set $\deltabox (\langle \overline \mu, \varphi \rangle) = f$ 
if and only if 
$\mu$ is a modulus of continuity of $f$ and 
for every $n \in \Nset$ and $u \in \Dset$, 
we have $v := \varphi (0 ^n, u) \in \Dset$ and 
$\lvert f (\lsem u \rsem) - \lsem v \rsem \rvert < 2 ^{-n}$
(the string~$v$ may have to have leading $0$s padded 
in order to make $\varphi$ regular---but 
this need for padding is inconsequential,
see the penultimate paragraph of Section~\ref{subsection: second-order polynomial};
in what follows, 
we omit this padding in the description of algorithms). 
To see that $\deltabox (\varphi)$ is well-defined, 
suppose that the above condition holds for two real functions $f$ and $f'$. 
Let $t \in [0, 1]$ be arbitrary. 
Then for each $n \in \Nset$, 
there is $u \in \Dset$ 
with $\lvert t - \lsem u \rsem \rvert \leq 2 ^{-\mu (n)}$
and thus 
\begin{align}
  \lvert f (t) - f' (t) \rvert 
&
 \leq 
   \lvert f (t) - f (\lsem u \rsem) \rvert 
  +
   \lvert f (\lsem u \rsem) - \lsem \varphi (0 ^n, u) \rsem \rvert 
\\
\notag
& \qquad {}
  +
   \lvert f' (\lsem u \rsem) - \lsem \varphi (0 ^n, u) \rsem \rvert 
  +
   \lvert f' (t) - f' (\lsem u \rsem) \rvert 
\\
\notag
&
 \leq 
   2 ^{-n} + 2 ^{-n} + 2 ^{-n} + 2 ^{-n}
 =
  2 ^{-n + 2}. 
\end{align}
Since $n \in \Nset$ was arbitrary, $f (t) = f' (t)$. 
Since $t \in [0, 1]$ was arbitrary, $f = f'$. 


Recall that the only reason 
that a real number can require long $\rhoreal$-names 
was having a large absolute value. 
In contrast, 
functions in $\classC [0, 1]$ 
may have long $\deltabox$-names for two possible reasons: 
having big values, 
and having a big modulus of continuity. 

The following lemma says that 
the complexity of $\deltabox$-names of $f \in \classC [0, 1]$ matches 
the complexity of $f$ that was discussed
in Section~\ref{subsection: applications, real numbers}
using the representation~$\rhoreal$: 

\begin{lemma}[{\cite[Corollary~2.21]{ko91:_comput_compl_of_real_funct}}]
 \label{lemma: box polytime iff polytime name}
A function in $\classC [0, 1]$ is polynomial-time 
(resp.\ polynomial-space) computable
if and only if 
it is in $\deltabox$-$\classFP$ 
(resp.\ $\deltabox$-$\classFPSPACE$). 
\end{lemma}

\begin{lemma}
\label{lemma: pspace-complete names}
A function in $\classC [0, 1]$ is 
$\classPSPACE$-complete in the sense of 
\cite[Section~2.2]{kawamura10:_lipsc_contin_ordin_differ_equat}
if it is $\deltabox$-$\classFPSPACE$-$\redonem$-complete
and has a polynomial modulus of continuity. 
\end{lemma}

\begin{proof}
Suppose that $f \in \classC [0, 1]$ is 
$\deltabox$-$\classFPSPACE$-$\redonem$-complete
and has a polynomial modulus of continuity~$\mu$. 
Then for any $A \in \classPSPACE$ 
there are polynomial-time functions $t$ and $r$ 
that satisfy the first picture of Figure~\ref{figure: type-one reductions}
for any $B \in \deltabox ^{-1} [f]$---and thus for 
any $B$ of the form $\langle \overline \mu, \varphi \rangle$
(that is, those with this particular polynomial~$\mu$ in the first component). 
A query to $B$ can ask either a value of $\mu$ or a value of $\varphi$, 
but $\mu$ is just a polynomial, 
so we may assume that $t$ only asks a query of form ``$\varphi (0 ^n, v)$?''. 
Thus, given $u$, 
an answer in $A [u]$ can be computed by $r$ from 
$u$ and a $2 ^{-n}$-approximation of $f (\lsem v \rsem)$. 
This implies that $f$ is 
$\classPSPACE$-complete. 
\end{proof}

The representation~$\deltabox$ of $\classC [0, 1]$ 
may look somewhat arbitrary at first sight. 
Here we present a property of $\deltabox$ 
that seems to make it a ``natural'' representation. 
Define the function $\OpApply \tcolon \classC [0, 1] \times [0, 1] \tto \Rset$ by 
$\OpApply (f, x) = f (x)$. 
The following lemma says that 
(the $\equivT$-equivalence class of) 
the representation~$\deltabox$ is 
the least informative representation of $\classC [0, 1]$
that makes $\OpApply$ efficiently computable
(see Section~\ref{subsubsection: translation and equivalence}
for the definitions of $\translT$ and $\equivT$). 
The proof will appear in a forthcoming paper. 

\begin{lemma}
\label{lemma: universality}
Let $\delta$ be any representation of $\classC [0, 1]$. 
Then $
\OpApply \in ([\delta, \rhoreal \restrictcodom{[0, 1]}], \rhoreal) \text- \classFPtwo
$ if and only if $\delta \translT \deltabox$. 
\end{lemma}

The above definitions and lemmas extend to 
some well-behaved compact domains other than $[0, 1]$
(we keep writing $\deltabox$ by abuse of notation). 
To discuss the complexity of the operator $\OpLipIVP$
(Section~\ref{subsubsection: non-uniform results}), 
we define a representation~$\deltalip$ of 
the space $\classCL [[0, 1] \times [-1, 1]]$ of Lipschitz continuous functions 
by setting $\deltalip (\langle \varphi, 0 ^L \rangle) = g$ if and only if 
$\varphi$ is a $\deltabox$-name of $g$ and 
$L \in \Nset$ satisfies \eqref{equation: lipschitz}
(regard the string $0 ^L$ as the constant function whose value is $0 ^L$). 

\subsubsection{Complexity of the operator that solves differential equations}
Now we can formulate the uniform version of 
Theorems \ref{theorem: ivp upper non-uniform}
and \ref{theorem: ivp lower non-uniform} as follows
(a proof will be given shortly). 

\begin{theorem}
 \label{theorem: ivp uniform}
$\OpLipIVP$ is 
$(\deltalip, \deltabox)$-$\classFPSPACEtwo$-$\redtwom$-complete. 
\end{theorem}

As corollaries, 
we have 
Theorem~\ref{theorem: ivp upper non-uniform} 
by 
Lemmas \ref{lemma: maps computable represented}.%
\ref{enumi: maps computable represented, functional}
and \ref{lemma: box polytime iff polytime name}, 
and 
Theorem~\ref{theorem: ivp lower non-uniform} 
by 
Lemmas 
\ref{lemma: maps complete represented}.%
\ref{enumi: maps complete represented, functional} 
and \ref{lemma: pspace-complete names}. 

The following weaker version of Theorem~\ref{theorem: ivp uniform}, 
stated with the stronger reduction $\redtwoT$ from
Definition~\ref{definition: many-one reduction}.\ref{enumi: Weihrauch reduction}, 
is slightly easier to prove (see the end of the section): 

\begin{corollary}
 \label{corollary: ivp uniform, strong reduction}
$\OpLipIVP$ is 
$(\deltalip, \deltabox)$-$\classFPSPACEtwo$-$\redtwoT$-complete. 
\end{corollary}

As we noted in Lemma~\ref{lemma: translation and completeness}, 
this is a more robust result in the sense that 
it is invariant under replacing representations to $\equivT$-equivalent ones. 
A drawback is that Corollary~\ref{corollary: ivp uniform, strong reduction} 
does not directly yield 
Theorem~\ref{theorem: ivp lower non-uniform}, 
because Lemma~\ref{lemma: pspace-complete names} requires 
$\classFPSPACE$-$\redonem$-completeness, 
whereas replacing $\redtwom$ by $\redtwoT$ in the assumption of 
\ref{lemma: maps complete represented}.%
\ref{enumi: maps complete represented, functional} 
also changes $\redonem$ to $\redoneT$ in the conclusion. 
We can still obtain 
Corollary~\ref{corollary: ivp lower non-uniform, weak}. 

The rest of the section is devoted to the proof of 
Theorem~\ref{theorem: ivp uniform}. 
The positive part ($\OpLipIVP \in (\deltalip, \deltabox)$-$\classFPSPACEtwo$) 
will be verified by 
checking that the proof of Theorem~\ref{theorem: ivp upper non-uniform} 
can be made uniform. 
For the hardness, 
we need to modify slightly the construction in~%
the original proof of Theorem~\ref{theorem: ivp lower non-uniform}
(this modification is not needed if we only 
want Corollary~\ref{corollary: ivp uniform, strong reduction}). 

\begin{proof}%
[Proof of Theorem~\ref{theorem: ivp uniform}, computability]
Given a $\deltalip$-name $\langle \overline \mu, \varphi, 0 ^L \rangle$ of $g$, 
we need to find a 
$\deltabox$-name $\langle \overline \nu, \psi \rangle$ of $h = \OpLipIVP (g)$. 
Recall that $\mu$ is a modulus of continuity of $g$, and $
 \lvert \lsem \varphi (0 ^q, u, v) \rsem - f (\lsem u \rsem, \lsem v \rsem) \rvert
<
 2 ^{-q}
$ for each $u$, $v \in \Dset$ 
(such that $(\lsem u \rsem, \lsem v \rsem) \in [0, 1] \times [-1, 1]$). 

It is easy to find a modulus of continuity $\nu$ of $h$: 
let $\nu (n) = n + M$, 
where $M \in \Nset$ is any number such that 
the values of $g$ always stay in $[-2 ^M, 2 ^M]$. 
For example, $
M = \lceil \log _2 (\lvert \lsem \varphi (\varepsilon, +0/1, +0/1) \rsem \rvert + 1 + 2 ^{\mu (0)}) \rceil
$. 

\def\Deltat{\varDelta t}

To obtain $\psi$, 
we apply the forward Euler method with step size $2 ^{-p}$ 
to the approximation of $g$ with precision $2 ^{-q}$ 
(we will specify $p$ and $q$ shortly). 
That is, 
we define an approximation $\tilde h _{p, q} \in \classC [0, 1]$ of $h$ 
by letting $
\tilde h _{p, q} (0) = 0
$ and then defining $
\tilde h _{p, q} 
$ on $[2 ^{-p} T, 2 ^{-p} (T + 1)]$, 
for each $T = 0, \ldots, 2 ^p - 1$ inductively, 
to be linear 
with slope approximately $g (2 ^{-p} T, h (2 ^{-p} T))$: formally, 
\begin{equation} 
 \label{equation: Euler step}
  \tilde h _{p, q} (2 ^{-p} T + \Deltat) 
 =
  \tilde h _{p, q} (2 ^{-p} T) + \Deltat \lsem \varphi (0 ^q, u, v) \rsem, 
\qquad
  0 \leq \Deltat \leq 2 ^{-p}, 
\end{equation}
for some $u$, $v \in \Dset$ 
with $\lsem u \rsem = 2 ^{-p} T$ and $\lsem v \rsem = \tilde h _{p, q} (2 ^{-p} T)$. 
Obviously, we can compute such a function $\tilde h _{p, q}$ 
in space polynomial in $p$ and $q$ 
in the sense that there is 
$\mathit{Euler} \in \classFPSPACEtwo$ such that $
\lsem \mathit{Euler} (\varphi) (0 ^p, 0 ^q, u) \rsem = \tilde h _{p, q} (\lsem u \rsem)
$
for every $u \in \Dset$. 

\newcommand{\e}{\mathrm e}

Let $\psi (0 ^n, u) = \mathit{Euler} (\varphi) (0 ^p, 0 ^q, u)$, 
where $p = \max \{\mu (n + 8 L), n + 8 L + M\}$ and 
$q = n + 8 L$. 
We claim that $\langle \overline \nu, \psi \rangle$ is a $\deltabox$-name of $h$ 
(this proves the desired $\classFPSPACEtwo$-computability, 
since $p$ and $q$ are bounded polynomially in 
$\lvert \varphi \rvert$, $\mu$ and $n$, $L$). 
This means that $
\lvert \tilde h _{p, q} (t) - h (t) \rvert \leq 2 ^{-n}
$ for any $t \in [0, 1]$. 
More strongly, 
we prove, by induction on $T = 0$, \ldots, $2 ^p - 1$, that
\begin{equation}
\label{equation: Euler induction hypothesis}
\lvert \tilde h _{p, q} (t) - h (t) \rvert \leq 2 ^{-n} \e ^{4 L (t - 1)} 
\end{equation} 
for all $t \in [2 ^{-p} T, 2 ^{-p} (T + 1)]$. 
We may assume \eqref{equation: Euler induction hypothesis} for $t = 2 ^{-p} T$
as the induction hypothesis. 
The approximate value $\tilde h _{p, q} (2 ^{-p} T + \Deltat)$ is 
defined by \eqref{equation: Euler step}, 
whereas the true solution~$h$ satisfies
\begin{equation}
 \label{equation: Euler step true}
 h (2 ^{-p} T + \Deltat) 
= 
  h (2 ^{-p} T) 
 + 
  \int _{2 ^{-p} T} ^{2 ^{-p} T + \Deltat} 
   g \bigl( \tau, h (\tau) \bigr) \, \mathrm d \tau. 
\end{equation}
The error added by this approximation is 
\begin{equation}
 \label{equation: Euler error estimate}
  \biggl|
   \Deltat \lsem \varphi (0 ^q, u, v) \rsem 
  - 
   \int _{2 ^{-p} T} ^{2 ^{-p} T + \Deltat} 
    g \bigl( \tau, h (\tau) \bigr) \, \mathrm d \tau 
  \biggr|
 \leq
  4 L 2 ^{-n} \e ^{4 L (2 ^{-p} T - 1)} \Deltat, 
\end{equation}
because 
\begin{align}
&
   \bigl| 
    \lsem \varphi (0 ^q, u, v) \rsem 
   - 
    g \bigl( \tau, h (\tau) \bigr)
   \bigr| 
\\ 
\notag
& \quad
 \leq
   \bigl| 
    \lsem \varphi (0 ^q, u, v) \rsem 
   - 
    g (\lsem u \rsem, \lsem v \rsem)
   \bigr| 
  +
   \bigl| 
    g (\lsem u \rsem, \lsem v \rsem)
   - 
    g (\tau, \lsem v \rsem)
   \bigr| 
  +
   \bigl| 
    g (\tau, \lsem v \rsem)
   - 
    g \bigl( \tau, h (\tau) \bigr)
   \bigr| 
\\ 
\notag
& \quad
 \leq
   2 ^{-q}
  + 
   2 ^{-n - 8 L}
  +
   L 
    \lvert \lsem v \rsem - h (\tau) \rvert
\\ 
\notag
& \quad
 \leq 
   2 ^{-n - 8 L}
  + 
   2 ^{-n - 8 L}
  + 
   L \bigl( 
    \lvert \lsem v \rsem - h (2 ^{-p} T) \rvert + \lvert h (2 ^{-p} T) - h (\tau) \rvert
   \bigr)
\\ 
\notag
& \quad
 \leq 
   2 ^{-n - 8 L}
  + 
   2 ^{-n - 8 L}
  + 
   L \bigl( 2 ^{-n} \e ^{4 L (2 ^{-p} T - 1)} + 2 ^{-p} 2 ^M \bigr) 
\\ 
\notag
& \quad
 \leq 
   L \bigl( 2 ^{-n - 8 L} + 2 ^{-n - 8 L} + 2 ^{-n} \e ^{4 L (2 ^{-p} T - 1)} + 2 ^{-n - 8 L} \bigr) 
 \leq
  4 L 2 ^{-n} \e ^{4 L (2 ^{-p} T - 1)}, 
\end{align}
where the second, third and fifth inequalities 
come from $
p \geq \mu (n + 8 L)
$, $
q \geq n + 8 L
$, $
p \geq M + n + 8 L
$, respectively. 
Using \eqref{equation: Euler error estimate} 
and the induction hypothesis, 
we compare \eqref{equation: Euler step} and \eqref{equation: Euler step true} 
to obtain
\begin{multline}
  \bigl| \tilde h _{p, q} (2 ^{-p} T + \Deltat) - h (2 ^{-p} T + \Deltat) \bigr| 
 \leq 
  2 ^{-n} \e ^{4 L (2 ^{-p} T - 1)} + 4 L 2 ^{-n} \e ^{4 L (2 ^{-p} T - 1)} \Deltat
\\
 =
  2 ^{-n} \e ^{4 L (2 ^{-p} T - 1)} (1 + 4 L \Deltat)
 \leq
  2 ^{-n} \e ^{4 L (2 ^{-p} T - 1)} \e ^{4 L \Deltat} 
 =
  2 ^{-n} \e ^{4 L (2 ^{-p} T - 1 + \Deltat)}, 
\end{multline}
as desired. 
\end{proof}

For the hardness, 
the core part of the proof can be done by relativizing the argument for 
the non-uniform version~%
\cite{kawamura10:_lipsc_contin_ordin_differ_equat}. 
Since the proof was by reduction from the problem~$\textsc{qbf}$, 
we use the relativized version $\textsc{qbf} ^2$ from 
Lemma~\ref{lemma: pspacetwo-complete}. 
Starting from $\textsc{qbf} ^2$, 
we follow the construction in
\cite[Lemma~4.1]{kawamura10:_lipsc_contin_ordin_differ_equat}, 
which uniformizes and yields the following. 
Let $\iota _{\varSigma ^*}$ be the representation of $\varSigma ^*$ 
which encodes a finite string $u$ by the constant function with value $u$. 
Let $\Lambda$ be the set of non-decreasing functions from $\Nset$ to $\Nset$, 
and let $\iota _{\Lambda}$ be its representation 
defined by $\iota _{\Lambda} (\varphi) = \overline \varphi$. 

\begin{lemma}
 \label{lemma: main lemma of paris}
Let $L \in \classPSPACEtwo$. 
Then there are a second-order polynomial~$P$ and a 
function $
 G
\in
 (
  [
   \id, \iota _{\Lambda}, \iota _{\varSigma ^*}, 
   \rhoreal \restrictcodom{[0, 1]}, \rhoreal \restrictcodom{[-1, 1]}
  ], 
  \rhoreal 
 )
\text-
\classFPtwo
$ such that 
for each $\varphi \in \dom L$, $\lambda \in \Lambda$ 
and $u \in \varSigma ^*$, 
the function 
$g ^{\varphi, \lambda} _u \tcolon [0, 1] \times [-1, 1] \tto \Rset$
defined by $g ^{\varphi, \lambda} _u (t, y) = G (\varphi, \lambda, u, t, y)$
satisfies
\begin{enumerate}
 \item \label{enumi: boundary}
  $g ^{\varphi, \lambda} _u (0, y) = g ^{\varphi, \lambda} _u (1, y) = 0$ for all $y \in [-1, 1]$; 
 \item \label{enumi: Lipschitz}
  $
   \bigl\lvert g ^{\varphi, \lambda} _u (t, y _0) - g ^{\varphi, \lambda} _u (t, y _1) \bigr\rvert 
  \leq 
   2 ^{-\lambda (\lvert u \rvert)} \lvert y _0 - y _1 \rvert 
  $ for any $
t \in [0, 1]$ and $y _0, y _1 \in [-1, 1]
  $; 
 \item \label{enumi: output}
  $g ^{\varphi, \lambda} _u (t, y) \in \dom \OpLipIVP$, and 
  $h ^{\varphi, \lambda} _u := \OpLipIVP \bigl( g ^{\varphi, \lambda} _u \bigr)$ satisfies
  $
 h ^{\varphi, \lambda} _u (1) 
= 
 2 ^{-P (\lvert \varphi \rvert, \lambda) (\lvert u \rvert)} \cdot L (\varphi) (u)
  $. 
\end{enumerate}
\end{lemma}

\begin{proof}%
[Proof of Theorem~\ref{theorem: ivp uniform}, hardness]
Let $F \in \classFPSPACEtwo$. 
We need to show that $F \redtwom \deltabox ^{-1} \circ \OpLipIVP \circ \deltalip$. 
We may assume that $F$ is a total function and 
that there is a second-order polynomial~$Q$ such that
$F (\varphi) (v)$ has 
length exactly $Q (\lvert \varphi \rvert) (\lvert v \rvert)$ 
for all $\varphi \in \Bset$ and $v \in \varSigma ^*$. 
There is $L \in \classPSPACEtwo$ such that 
$L (\varphi) (v, 0 ^i)$ equals the $i$th symbol of $F (\varphi) (v)$
for any $\varphi \in \Bset$, $v \in \varSigma ^*$ and 
$i \in \{0, 1, \ldots, Q (\lvert \varphi \rvert) (\lvert v \rvert) - 1\}$. 
Apply Lemma~\ref{lemma: main lemma of paris} to this $L$ 
to obtain the $P$ and $G$, 
and let $g ^{\varphi, \lambda} _u$ and $h ^{\varphi, \lambda} _u$ be 
as in the Lemma. 

We define $s$ (of Definition~\ref{definition: many-one reduction}) 
by describing the real function 
$g = \deltalip (s (\varphi)) \in \classCL [[0, 1] \times [-1, 1]]$
for a given $\varphi$. 
It has Lipschitz constant~$1$. 
It will be straightforward to check that
a $\deltabox$-name 
(and hence a $\deltalip$-name) of $g$ can be $\classFPtwo$-computed from $\varphi$. 
We write $g _u$ and $h _u$ for the 
$g ^{\varphi, \lambda} _u$ and $h ^{\varphi, \lambda} _u$ corresponding to 
this given $\varphi$ and $\lambda (k) = 3 k + 2$. 

For each binary string~$v$, let 
\begin{align}
  c _v 
&
 =
   1
  -
   \frac{1}{2 ^{\lvert v \rvert}} 
  + 
    \frac{2 \overline v + 1}{2 ^{2 \lvert v \rvert + 2}}, 
&
  l ^\mp _v
&
 = 
  c _v \mp \frac{1}{2 ^{2 \lvert v \rvert + 2}}, 
\end{align}
where $\overline v \in \{0, \dots, 2 ^{\lvert v \rvert} - 1\}$ means 
$v$ interpreted as an integer in binary notation. 
This divides $[0, 1)$ into intervals $[l ^- _v, l ^+ _v]$
indexed by $v \in \{0, 1\} ^*$. 
We further divide the left half $[l ^- _v, c _v]$ into 
$Q (\lvert \varphi \rvert) (\lvert v \rvert) + 1$ subintervals 
$[l _{v, 0}, l _{v, 1}]$, $[l _{v, 1}, l _{v, 2}]$, \ldots, 
$[l _{v, Q (\lvert \varphi \rvert) (\lvert v \rvert) - 1}, l _{Q (\lvert \varphi \rvert) (\lvert v \rvert)}]$, 
$[l _{v, Q (\lvert \varphi \rvert) (\lvert v \rvert)}, c _v]$, where 
\begin{align}
 l _{v, i} & = c _v - \frac{1}{2 ^{2 \lvert v \rvert + 2} 2 ^i}, 
&
 i = 0, 1, \ldots, Q (\lvert \varphi \rvert) (\lvert v \rvert). 
\end{align}
On each strip $[l _{v, i}, l _{v, i + 1}] \times [-1, 1]$, 
we define $g$ by putting the copies of $g _{(v, 0 ^i)}$
as in Figure~\ref{figure: ivp widget}. 
\begin{figure}
\begin{center}
\includegraphics[scale=1.0]{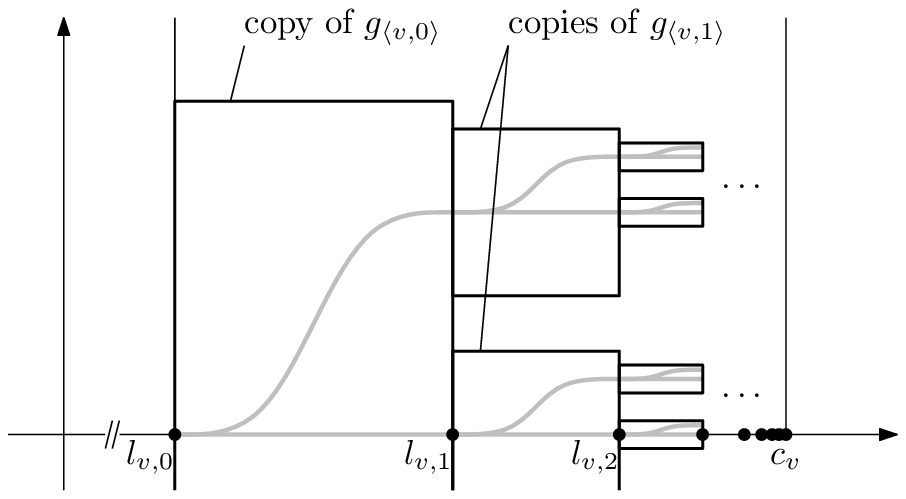}
\caption{Widget for reducing $\classFPSPACEtwo$ to $\OpLipIVP$.}
\label{figure: ivp widget}
\end{center}
\end{figure}
Precisely, 
\begin{equation}
 \label{equation: vector}
  g \biggl( 
   l _{v, i} + \frac{t}{2 ^{2 \lvert v \rvert + 2} 2 ^{i + 1}},
   \frac{2 m + (-1 ) ^m y}{2 ^{\gamma (\lvert v \rvert, i)}} 
  \biggr)
 = 
  \frac{2 ^{2 \lvert v \rvert + 2} 2 ^{i + 1}}{2 ^{\gamma (\lvert v \rvert, i)}} g _{(v, 0 ^i)} (t, y)
\end{equation}
for each $t \in [0, 1]$ and $m \in \Nset$, $y \in [-1, 1]$, 
where the polynomial~$\gamma$ is defined by $
 \gamma (\lvert v \rvert, 0) 
= 
 2 \lvert v \rvert + 3 
$ and $
 \gamma (\lvert v \rvert, i + 1) 
= 
  \gamma (\lvert v \rvert, i) 
 + 
  P (\lvert \varphi \rvert, \lambda) (\lvert (v, 0 ^i) \rvert) + 2 
$. 
On the last strip 
$[l _{v, Q (\lvert \varphi \rvert) (\lvert v \rvert)}, c _v] \times [-1, 1]$, 
we define $g$ to be constantly $0$. 
On the right half $[c _v, l ^+ _v]$, 
we define $g$ symmetrically: $g (l ^+ - t, y) = -g (l ^- + t, y)$
for $0 \leq t \leq 1 / 2 ^{2 \lvert v \rvert + 2}$. 
Because of this symmetry, 
the function $h := \OpLipIVP (g)$ takes value $0$
at each $l ^{\mp} _v$, 
and it can be verified, 
using (\ref{enumi: output}) of 
Lemma~\ref{lemma: main lemma of paris}, 
that 
for $i = 0$, \ldots, $Q (\lvert \varphi \rvert) (\lvert v \rvert)$ inductively, 
\begin{equation*}
 h (l _{v, i})
=
 \sum _{j = 0} ^{i - 1}
  \frac{h _{(v, 0 ^j)} (1)}{2 ^{\gamma (\lvert v \rvert, j)}}
=
 \sum _{j = 0} ^{i - 1}
  \frac{L (\varphi) (v, 0 ^j)}{2 ^{\gamma (\lvert v \rvert, j)} 2 ^{P (\lvert \varphi \rvert, \lambda) (\lvert (v, 0 ^i) \rvert)}} 
=
 \sum _{j = 0} ^{i - 1}
  \frac{4 L (\varphi) (v, 0 ^j)}{2 ^{\gamma (\lvert v \rvert, j + 1)}}. 
\end{equation*} 
In particular, the number $
 h (c _v) 
= 
 h (l _{v, Q (\lvert \varphi \rvert) (\lvert v \rvert)}) 
= 
 \sum _{j = 0} ^{Q (\lvert \varphi \rvert) (\lvert v \rvert) - 1}
  4 L (\varphi) (v, 0 ^j) / 2 ^{\gamma (\lvert v \rvert, j + 1)}
$ contains 
the values $L (\varphi) (v, 0 ^j)$ 
for all $j < Q (\lvert \varphi \rvert) (\lvert v \rvert)$, 
from which we can recover $F (\varphi) (v)$. 
The reducing functions $r$ and $t$ 
(of Definition~\ref{definition: many-one reduction}) 
perform this lookup. 
That is, $
t (\varphi) (v) = (0 ^{\gamma (\lvert v \rvert, Q (\lvert \varphi \rvert) (\lvert v \rvert))}, w)
$ with $\lsem w \rsem = c _v$, 
and $r (\varphi)$ is the function that, 
given the encoding of (an approximation of) $h (c _v)$, 
extracts the value $F (\varphi) (v)$. 
\end{proof}

In \cite[Theorem~3.2]{kawamura10:_lipsc_contin_ordin_differ_equat}, 
the non-uniform version of Lemma~\ref{lemma: main lemma of paris} 
was used to construct 
a function that proved 
Theorem~\ref{theorem: ivp lower non-uniform}. 
We needed a different construction, 
because for our Theorem~\ref{theorem: ivp uniform} 
(with the reduction $\redtwom$), 
we needed to get the values $L (\varphi) (v, 0 ^j)$ 
for all $j < Q (\lvert \varphi \rvert) (\lvert v \rvert)$ 
in one query. 
For Corollary~\ref{corollary: ivp uniform, strong reduction}
(with the reduction $\redtwoT$), 
we are allowed to make many queries, so 
the straightforward uniformization 
(without stacking the copies of $g _{(v, 0 ^i)}$ vertically)
would have worked. 

\section{Summary and future work}

\begin{itemize}
\item
To discuss computational complexity in the framework of TTE, 
we replace $\varSigma ^\Nset$, the infinite strings, 
by $\Bset$, a class of functions from strings to strings. 
This is a generalization in two ways: 
these functions
\ref{enumi: predicate} can have values of arbitrary length, and
\ref{enumi: unary} take string arguments, rather than just unary strings. 
\item 
For time and space bounds we use second-order polynomials in the input size, 
which are defined in the way suggested by type-two complexity theory. 
We defined classes $\classPtwo$, $\classNPtwo$ 
and $\classFPtwo$, $\classFPSPACEtwo$. 
With a suitable notion of polynomial-time reductions, 
we can also define $\classNPtwo$- 
and $\classFPSPACEtwo$-completeness. 
Formulating other classes is left for future work. 
\item 
To apply this to problems involving real numbers, 
we introduced representations 
$\rhoreal$, $\psiclosed$ and $\deltabox$
of real numbers, sets and functions. 
Both aspects \ref{enumi: predicate} and \ref{enumi: unary} of our generalization 
were useful. 
With respect to these representations, 
we showed that 
taking the convex hull of a set is $\classNPtwo$-complete, and that 
solving a Lipschitz continuous ordinary differential equation 
is $\classFPSPACEtwo$-complete. 
These are uniform versions of what have been known non-uniformly, 
and tell us more about the hardness of numerical problems in practice. 
An interesting direction for further investigation is to ask
which other known non-uniform results about operators 
do (or do not) uniformize. 
One can also look at known computability results and 
ask whether analogous statements hold true 
for time- or space-bounded classes. 
\end{itemize}

\small
\section*{Acknowledgements}
We thank 
Vasco Brattka, 
Kaveh Ghasemloo, 
Ken Jackson, 
Toni Pitassi, 
Bill Weiss
and anonymous referees 
for comments on this and related manuscripts
which helped improve the presentation. 
We also thank 
Keiko Imai and Yu Muramatsu for providing 
the image of the trisector curves (Figure~\ref{figure: trisector}). 
During this research, 
the first author was supported by
the Nakajima Foundation and by the
Grant-in-Aid for Scientific Research (Kakenhi) 23700009; 
both authors were supported by 
the Natural Sciences and Engineering Research Council of Canada.

\end{document}